\documentclass[runningheads]{llncs}
\usepackage{color} \usepackage[utf8]{inputenc}
\usepackage{amsfonts}
\usepackage{amssymb}% http://ctan.org/pkg/amssymb
\usepackage{pifont}% http://ctan.org/pkg/pifont
\usepackage{stmaryrd}
\usepackage{lmodern}
\usepackage{multirow}

\usepackage{nccmath}
\usepackage{blkarray} 

\usepackage{float}
\usepackage[color]{changebar}
\cbcolor{red}
\usepackage{tikz,pgfplots}
\usepackage{wrapfig}
\usepackage{lipsum}

\usepackage{mathtools}
\DeclarePairedDelimiter\ceil{\lceil}{\rceil}
\DeclarePairedDelimiter\floor{\lfloor}{\rfloor}
\usepackage{listings}
\usepackage{xcolor}
\usepackage[breakable, theorems, skins]{tcolorbox}

\usepackage{graphicx}
\usepackage[T1]{fontenc}
\usepackage{mathpartir}
\usepackage[ruled,vlined,linesnumbered]{algorithm2e}
\definecolor{codegreen}{rgb}{0,0.6,0}
\definecolor{codegray}{rgb}{0.5,0.5,0.5}
\definecolor{codepurple}{rgb}{0.58,0,0.82}
\definecolor{backcolour}{rgb}{0.95,0.95,0.92}

\newcommand{\eguess}{::}
\newcommand{\irefine}{::}

\newcommand{\numAllBenchmarks}{77\xspace}
\newcommand{\numMoreBenchmarks}{5\xspace}

\newcommand{\numBenchmarksSolvedBySynquid}{75\xspace}

\newcommand{\share}{\curlyveedownarrow}

\usepackage[framemethod=tikz]{mdframed}

\lstdefinestyle{mystyle}{
	commentstyle=\color{codegreen},
	keywordstyle=\color{black},
	numberstyle=\tiny\color{codegray},
	stringstyle=\color{codepurple},
	basicstyle=\ttfamily\footnotesize,
	breakatwhitespace=false,         
	breaklines=true,                 
	captionpos=b,                    
	keepspaces=true,                 
	numbers=left,                    
	numbersep=5pt,                  
	showspaces=false,                
	showstringspaces=false,
	showtabs=false,                  
	tabsize=2
}

\newenvironment{mybox}[1][gray!20]{
	\begin{tcolorbox}[   %% Adjust the following parameters at will.
		breakable,
		left=0pt,
		right=0pt,
		top=0pt,
		bottom=-1pt,
		colback=#1,
		colframe=#1,
		width=1\dimexpr\textwidth\relax,
		boxsep=4pt,
		arc=0pt,outer arc=0pt,
		]
	}{
	\end{tcolorbox}
}

\newcounter{finding}%[section]

\newcommand {\john}[1]{{\color{violet}\sf{J: #1}\normalfont}}
\newcommand {\checkcolor}[1]{{\color{red}#1\normalfont}}
\definecolor{mediumelectricblue}{rgb}{0.1, 0.1, 0.69}
\newcommand {\generatecolor}[1]{{\color{mediumelectricblue}#1\normalfont}}

\newcommand{\mypar}[1]{\vspace{0.5mm}\noindent  \textbf{\textit{#1}}\quad}
\newcommand{\config}[1]{\langle #1\rangle^\#}
\newcommand{\concreteConfig}[1]{\langle #1\rangle}
\newcommand{\sem}[1]{\llbracket #1\rrbracket}
\def\fbx#1{\vbox{\hbox{\hbox{#1}\setbox0\lastbox\copy0\kern\fboxsep\vrule width\fboxrule depth\dimexpr \fboxsep+\dp0\relax}%
		\hrule height\fboxrule}}

\newcommand{\bound}{\sqsubset}
\def\name{{\normalfont\textsc{SynPlexity}}\xspace}
\def\nameAuxAlgo{{\normalfont\textsc{SynAuxRef}}\xspace}
\def\tool{\name}
\def\resyn{{\normalfont\textsc{ReSyn}}\xspace}
\def\synquid{{\normalfont\textsc{Synquid}}\xspace}

\font\btt=rm-lmtk10
\lstdefinelanguage{OCaml}{
	keywords={fix, Expression, match, with, let,Div_2, Number, in,Mult, if, then,else,ite},
	keywordstyle=\btt,
	ndkeywords={class, export, boolean, throw, implements, import, this},
	ndkeywordstyle=\color{black}\bfseries,
	identifierstyle=\color{black},
	sensitive=false,
	comment=[l]{//},
	morecomment=[s]{/*}{*/},
	commentstyle=\color{purple}\ttfamily,
	stringstyle=\color{red}\ttfamily,
	morestring=[b]',
	morestring=[b]"
}

\newcommand{\subsubsubsection}[1]{\vspace{2pt plus 1pt minus 1pt}\noindent{\bf #1}}
\newcommand{\Omit}[1]{}

\usepackage{refs}

\newif\iffull
\fulltrue
\begin{document}
\title{Synthesis with Asymptotic Resource Bounds}
\author{Qinheping Hu, John Cyphert, Loris D'Antoni,   Thomas Reps}
\institute{University of Wisconsin-Madison, Madison, USA}

\maketitle              % typeset the header of the contribution
\begin{abstract}
		We present a method for synthesizing recursive functions that
		satisfy both a functional specification and an asymptotic resource
		bound.
		Prior methods for synthesis with a resource metric require the user to
		specify a \emph{concrete} expression exactly describing resource usage, whereas our method uses big-$O$
		notation to specify the \emph{asymptotic} resource usage.
		Our method can synthesize programs with complex resource bounds,
		such as a sort function that has complexity $O(n\log(n))$.
		
		\hspace{1.5ex}
		Our synthesis procedure uses a type system that is able
		to assign an asymptotic complexity to terms, and can track recurrence relations of functions.
		These typing rules are justified by  theorems used in
		analysis of algorithms, such as the Master Theorem and the
		Akra-Bazzi method.
		We implemented our method as an extension of prior type-based
		synthesis work.		
		Our tool, \name, was able to synthesize complex divide-and-conquer programs
		that cannot be synthesized by prior solvers.
\end{abstract}

\section{Introduction}
Program synthesis is the task of automatically finding programs
that meet a given behavioral specification, such as input-output
examples or complete formal specifications.
Most of the work on program synthesis has been devoted to qualitative synthesis, i.e., finding \emph{some} correct solution.
However, programmers often want more than just a correct solution---they
may want the program that is smallest, most likely, or most
efficient.
While there are some techniques for adding a quantitative \textit{syntactic}
objective in program synthesis \cite{hu2018syntax}---e.g.,
finding a smallest solution,  or a most likely solution with
respect to some distribution---little attention has been devoted to
quantitative \textit{semantic} objectives---e.g., synthesizing a program
that has a certain asymptotic complexity. 

Recently, Knoth et al.~\cite{knoth2019resource} studied the problem of
resource-guided program
synthesis, where the goal is to synthesize programs with limited resource usage. 
Their approach, which combines refinement-type-directed
synthesis \cite{polikarpova2016program} and automatic amortized
resource analysis (AARA) \cite{hoffmann2011multivariate},
is restricted to \textit{concrete} resource bounds,
where the user must specify the \emph{exact} resource usage of
the synthesized program as a \emph{linear} expression.
This limitation has two drawbacks:
(i) the user must have insights about the coefficients to put in the
supplied bound---which means that the user has to provide details
about the complexity of code that does not yet exist;
(ii) the limitation to linear bounds means that the user cannot specify
resource bounds that involve logarithms, such as $O(\log n)$ and
$O(n \log n)$, common in problems
based  on divide and conquer.

In this paper, we introduce \tool, a type-system paired with a
type-directed synthesis technique that addresses these issues.
In \tool, the user provides as input a refinement type that describes both
the functionality and the \textit{asymptotic} (big-$O$) resource
usage of a program.
For example, a user might ask \tool to synthesize an implementation of a sorting function with resource usage $O(n\log n)$,
where $n$ is the length of the input list.
As in prior work, \tool also takes as input a set of auxiliary functions
that the synthesized program can use.
\tool then uses a type-directed synthesis algorithm to search for a program
that has the desired functionality,  and satisfies the asymptotic resource bound.
\tool's synthesis algorithm uses a new type system that can reason about the asymptotic complexity of functions.
To achieve this goal, this type system uses two ideas.
\begin{enumerate}
  \item
    The type system uses \textit{recurrence relations}  instead of concrete resource potentials \cite{hoffmann2011multivariate}
    to reason about the asymptotic complexity of functions.
    For example, the recurrence relation $T(\emph{u})\leq 2T(\floor{\frac{\emph{u}}{2}})+O(\emph{u})$ denotes that 
    on an input of size $\emph{u}$, the function will perform at most two recursive calls
    on inputs of size at most $\floor{\frac{\emph{u}}{2}}$, and will use at most $O(\emph{u})$ resources
    outside of the recursive calls.\footnote{
      The recurrence relation above is one possible instantiation of the Master Theorem \cite[\S4.5 and \S4.6]{Book:CLR01};
      it can also be instantiated as $T(\emph{u})\leq 2T(\ceil{\frac{\emph{u}}{2}})+O(\emph{u})$.
      The type system makes use of certain templates for instantiating the algorithm-analysis theorems
      that we use.
      The use of templates means that the type system does not use all possible instantiations,
      but all instantiations used in the type system are valid ones.
    }
    For a given recurrence relation, our type system uses refinement types to guarantee that a function typed with this 
    recurrence relation performs 
    the correct number of recursive calls on parameters of the appropriate sizes.
  \item
    These typing rules are justified by classic theorems from the
    field of analysis of algorithms, such as the Master Theorem
    \cite{Book:CLR01}, the Akra-Bazzi method \cite{COA:AK98}, or
    C-finite-sequence analysis \cite{Book:KP11}.
\end{enumerate}

Gu{\'e}neau et al.\ observed that reasoning with $O$-notation
can be tricky, and exhibited a collection of plausible-sounding,
but flawed, inductive proofs \cite[\S2]{gueneau2018fistful}.
We avoid this  pitfall via \tool's type system, which establishes
whether a term satisfies a given recurrence relation.
\tool uses theorems
that connect the form of a recurrence relation---e.g., the number of
recursive calls, and the argument sizes in the subproblems---to its
asymptotic complexity.
In particular, the \tool type system does not encode inductive proofs
of the kind that Gu{\'e}neau et al.\ show can go astray.

\tool can 
synthesize functions with complexities that cannot be handled by existing type-directed tools~\cite{polikarpova2016program,knoth2019resource},
and compares   favorably with existing tools on their benchmarks.
Furthermore, for some domains, \tool's type system allows us to discover auxiliary functions automatically
(e.g., the split function of a merge sort), instead of requiring the user to
provide them.

\vspace{1mm}\noindent\emph{Contributions.}
The contributions of our work are as follows:
\begin{itemize}
  \item
    A type system that uses refinement types to check
    whether a program satisfies a recurrence relation over a specified resource  (\sectref{typeSystem}).
  \item
    A type-directed algorithm that uses our type system to synthesize functions
    with given resource bounds (\sectref{synthesisAlgorithm}, \sectref{extension}).
  \item
    \tool, an implementation of our algorithm that, unlike prior tools, can synthesize programs with desired asymptotic complexities (\sectref{eval}).
\end{itemize}
Complete proofs and details of the type system can be found in the appendices.
%%% Local Variables: 
%%% mode: latex
%%% TeX-master: "main_cav.tex"
%%% End: 

\section{Overview}
\label{Se:Overview}

In this section, we illustrate the main components of our algorithm
through an example.
Consider the problem of synthesizing a function \texttt{prod}
that implements the multiplication of two natural numbers, $x$ and $y$.
We want an efficient solution whose time complexity is $O(\log x)$
with respect to the value of the first argument $x$.
In \sectref{typeSynthesis}, we show how existing type-directed
synthesizers solve this problem in the absence of a complexity-bound
constraint.
In \sectref{addBounds}, we illustrate how to specify asymptotic bounds
in type-directed synthesis problems.
In \sectref{fromBoundToRelation}, we show how the tracking of recurrence relations can be used to establish complexity bounds as well as guide the synthesis search.

\subsection{Type-Directed Synthesis}
\label{Se:typeSynthesis}

We first review one of the state-of-the-art type-directed synthesizers, \synquid, through the
aforementioned example---i.e., synthesizing a program \texttt{prod} that computes the product of two natural numbers.
In \synquid, the specification is given as a refinement type that describes the desired behavior of the synthesized function. 
We specify the behavior of $\texttt{prod}$ using the following refinement-type:
\[
  \texttt{prod}::\texttt{x:}\{\texttt{Int}~|~\emph{v}\ge0\}\to \texttt{y:}\{\texttt{Int}~|~\emph{v}\ge0\}\to\{\texttt{Int}\mid\emph{v}=\texttt{x}* \texttt{y}\}.
\]
Here the types of the inputs \texttt{x} and \texttt{y}, as well as the
return type of \texttt{prod} are refined with predicates.
The refinement $\{\texttt{Int}~|~\emph{v}\ge0\}$ declares \texttt{x} and \texttt{y} to be non-negative,
and the refinement  $\{\texttt{Int}\mid\emph{v}=\texttt{x}* \texttt{y}\}$ of the return type
declares the output value to be an integer that is equal to the product of the inputs \texttt{x} and \texttt{y}.
In addition to the specification, the synthesizer receives as input
some signatures of auxiliary functions it can use.
The specifications of auxiliary functions are also given as refinement types.
In our example, we have the following functions:
\begin{eqnarray*}
	\texttt{even}&::&\texttt{x:Int}\to\{\texttt{Bool}\mid\texttt{x mod}\ 2=0\}\qquad	\texttt{dec}::\texttt{x:Int}\to\{\texttt{Int}\mid \emph{v} = \texttt{x} - 1\}\\
	\texttt{double}&::&\texttt{x:Int}\to\{\texttt{Int}\mid \emph{v} = \texttt{x} + \texttt{x}\}\qquad\quad\hspace{0.8mm}
	\texttt{div2}::\texttt{x:Int}\to\{\texttt{Int}\mid \emph{v} = \floor{\frac{\texttt{x}}{2}}\}\\
	\texttt{plus}&::&\texttt{x:Int}\to\texttt{y:Int}\to\{\texttt{Int}\mid \emph{v}=\texttt{x} + \texttt{y}\}
\end{eqnarray*}

With the above specification and auxiliary functions, \synquid will output the implementation of \texttt{prod} shown in \eqref{prod}.
\begin{eqnarray}
\label{Eq:prod}
	\texttt{prod = }\lambda\texttt{x.}\lambda y.\ \texttt{\btt{if}} \texttt{ x==0 }\texttt{\btt{then}} \texttt{ x \btt else } \texttt{plus y (prod (dec x) y)}
\end{eqnarray}
\synquid uses a sophisticated type system to guarantee that the synthesized term has the desired type. 
Furthermore, \synquid uses its type system to prune the search space by only enumerating terms that can possibly be typed, and thus meet the specification.
Terms are enumerated in a top-down fashion, and appropriate specifications are propagated to sub-terms.   As an example, let us see how \synquid synthesizes the function body---an \texttt{if-then-else} term---in \eqref{prod}, which is of refinement type $\{\texttt{Int}~|~\emph{v}=\texttt{x}*\texttt{y}\}$. \synquid will first enumerate an integer term for the \texttt{then} branch---a variable term $\texttt{x}$. Then, with the \texttt{then} branch fixed, the condition guard must be refined by some predicate $\varphi$ under which the \texttt{then} branch (the term \texttt{x} refined by $\emph{v}=\texttt{x}$) fulfills the goal type $\{\texttt{Int}~|~\emph{v}=\texttt{x}*\texttt{y}\}$, i.e., $\forall \texttt{x},\texttt{y}\ge0.\varphi\wedge \emph{v}=\texttt{x}\implies\emph{v}=\texttt{x}*\texttt{y}$. With this constraint, \synquid identifies the term $\texttt{x}==0$ as the condition.
Finally, \synquid propagates the negation of the condition to the \texttt{else} branch---the \texttt{else} branch should be a term of type 
 	%$\{\texttt{Int}~|~\emph{v}=\texttt{x}*\texttt{y}\wedge \texttt{x}\ge0 \wedge \texttt{x}\ne0\}$ 
 	$\{\texttt{Int}~|~\emph{v}=\texttt{x}*\texttt{y}\}$
 	with the path condition $\texttt{x} \ne 0$---and enumerates the term \texttt{plus y (prod (dec x) y)} as the else branch, which has the desired type.
 
The program in \eqref{prod} is correct, but inefficient.
Let us count each call to an auxiliary function as one step;
and let $T(x)$ denote the number of steps in which the
program runs with input $x$.
The implementation in \eqref{prod} runs in $ \Theta(x)$
steps because  $T(x)$ satisfies  the recurrence
$T(x)= T(x-1)+2$, implying $T(x)\in \Theta(x)$. 
Because, \synquid does not provide a way to specify resource bounds,
such as $O(\log x)$; one cannot ask \synquid to find a more efficient implementation.   

\subsection{Adding Resource Bounds} 
\label{Se:addBounds}

In our tool, \tool, one can specify a synthesis problem with an
asymptotic resource bound, and can ask \tool to find an $O(\log x)$
implementation of \texttt{prod}.
To express this intent, the user needs to specify (1) the asymptotic
resource-usage bound the synthesized program should satisfy,
(2) the cost of each provided auxiliary function, and
(3) the size of the input to the program. 

\emph{Asymptotic Resource Bound.} We extend refinement types with resource annotations. The annotated refinement types are of the form $\langle\tau;\generatecolor{\alpha}\rangle$ where $\tau$ is a regular refinement type, and $\alpha$ is a resource annotation.
The following example asks the synthesizer to find a solution with the resource-usage bound $O(\log \emph{u})$:

\begin{eqnarray}
\texttt{prod}::\langle\texttt{x:}\{\texttt{Int}\mid\emph{v}\ge\nonumber 0\}\to\texttt{y:}\{\texttt{Int}\mid\emph{v}\ge\nonumber 0\}\to\{\texttt{Int}\mid\emph{v}=\texttt{x}* \texttt{y}\},
\generatecolor{O(\log \emph{u})}\rangle
\end{eqnarray}

\emph{Cost of Auxiliary Functions.}
The auxiliary functions supplied by  the user serve as the API
in terms of which the synthesized program is programmed.
Thus, the resource usage of the synthesized program is the sum
of the costs of all auxiliary calls made during execution.
	We allow users to assign a polynomial cost $O(\emph{u}^a)$, for some constant $a$, or a constant cost $O(1)$
	 to each auxiliary function. Here,  $\emph{u}$ is a free variable that represents the
	size of the problem on which the auxiliary function is called.

In the $\texttt{prod}$ example, all auxiliary functions are assigned constant cost, e.g., we give $\texttt{even}$ the signature
$\texttt{even}::\langle\texttt{x:Int}\to\{\texttt{Bool}\mid\texttt{x mod}\ 2=0\},\generatecolor{O(1)}\rangle$.

\emph{Size of Problems.}
The user needs to specify a size function, \texttt{size:}$\tau\to\texttt{Int}$, that maps inputs to their sizes, e.g., when synthesizing the sorting function for an input of type \texttt{list}, the size function can be $\lambda\texttt{l}.|l|$---the length of the input list. In the \texttt{prod} example, the size function is $\texttt{size}=\lambda \texttt{x}. \lambda \texttt{y}.\texttt{x}$.

\subsection{Checking Recurrence Relations}
\label{Se:fromBoundToRelation}
We extend
\synquid's
refinement-type system with resource annotations, so that the extended type system enforces the resource usage of terms. 
The idea of the type system is to check if the given function satisfies some recurrence relation. If so, it can infer that the function also satisfies the corresponding resource bound. For example, according to the Master Theorem \cite{bentley1980general}, if a function $f$ satisfies the recurrence relation $T(\emph{u})\le T(\floor{\frac{\emph{u}}{2}}) + O(1)$ where $\emph{u}$ is the size of the input, then the resource usage of $f$ is bounded by $O(\log\emph{u})$.
Checking if a function satisfies a given recurrence relation can be
performed by checking if the function contains appropriate recursive
calls---e.g., if a function contains one recursive call to a
sub-problem of half size, and consumes only a constant amount of
resources in its body, then it satisfies
$T(\emph{u})\le T(\floor{\frac{\emph{u}}{2}}) + O(1)$.

The following rule is an example of how we connect
recurrence annotations and resource bounds. 
	\begin{mathpar}
	{\inferrule*[right=\textsc{\ }]
		{			x:\tau_1,f:\tau_1\to\tau_2,\Gamma\vdash t::\langle\tau_2;(\generatecolor{[1,\floor{\frac{\emph{u}}{2}}]},\generatecolor{O(1)})\rangle}
		{ \Gamma\vdash\left(\texttt{fix}~f.~\lambda x.t\right)::\langle \tau_1\to\tau_2;\generatecolor{O(\log \emph{u})}\rangle}
	}
\end{mathpar}
The rule 
instantiates the Master Theorem example above.
Note that, the annotation  $(\generatecolor{[1,\floor{\frac{\emph{u}}{2}}]},\generatecolor{O(1)})$ 
states that the function body contains up to one recursive call to a
problem of size $\floor{\frac{\emph{u}}{2}}$, and the resource usage
in the body of $t$ (aside from calls to $f$ itself) is bounded by $\generatecolor{O(1)}$.
The rule states that if the function body $t$ of type $\tau_2$ contains one recursive call
to a sub-problem of size $\floor{\frac{\emph{u}}{2}}$,
then the function will be bounded by $\generatecolor{O(\log\emph{u})}$. 

The implementation of \texttt{prod} shown in \eqref{prod_log} runs in $O(\log x)$ steps. 
\begin{eqnarray}
\label{Eq:prod_log}
&\hspace{-40mm}\texttt{prod = }\lambda\texttt{x. }\lambda\texttt{y.\texttt{\btt if }x == 0 \texttt{\btt then }x	\texttt{\btt	else } }\\
&\hspace{-10mm}\texttt{\btt		if } \texttt{even x \texttt{\btt then } double (prod (div2 x) y) }\nonumber\\
&\hspace{22mm}\texttt{				\texttt{\btt	else } plus y (double (prod (div2 x) y))}\nonumber
\end{eqnarray}
To check that, \tool's type system counts the number of recursive calls along any path of the function. There are three paths (two nested  if-then-else terms) in the program, and at most one recursive call along each path. Also, one can check that the problem size of each recursive call is no more than $\floor{\frac{\texttt{x}}{2}}$.
For example, the recursive call $\texttt{prod (div2 x) y}$ calls to a problem with size $\texttt{div2 x}$, which is consistent with $\generatecolor{[1,\floor{\frac{\emph{u}}{2}}]}$, and $\emph{u}$ is $\texttt{x}$ because $\texttt{size}~\texttt{x}~\texttt{y} = \texttt{x}$.
In addition, the condition that the resource usage of the body is bounded by $O(1)$ is satisfied because only auxiliary functions with constant cost are called.

%%% Local Variables: 
%%% mode: latex
%%% TeX-master: "main_cav.tex"
%%% End: 

\section{The {\name} Type System}
\label{Se:typeSystem}
In this section, we present our type system. 
First, we give the surface language and the types, which extend the \synquid liquid-types framework with resource annotations (\sectref{syntax}). Then, we show the semantics of our language (\sectref{semantics}).
Finally, we present \name's type system  (\sectref{TypingRules}), which our synthesis algorithm uses to synthesize programs with desired resource bounds.

\subsection{Syntax and Types}
\label{Se:syntax}
\begin{figure}[t!]
	\[\hspace{-10mm}\arraycolsep=1.4pt
	\begin{array}{lrclrcl}
%	\text{Size, index}&   a,c &&\text{Non-negative integers}\\
%	\text{Integer expr.}& I::=&&x~|~c~|~I+I~|~I-I~|~I\div I~|~\cdots \\
	\text{Term}	& t::=&& e~|~b \\
	\text{E-term}	&	e::=&& x~|~c~|~\texttt{true}~|~\texttt{false}~|~ x \, e_1 \ldots e_n \\
\hspace{-12.9mm}\text{I-term}	\left\{ \begin{array}{l}
	\text{Branching term~}\\
	\\
	\text{Function term} 
	\end{array}\right. &\begin{array}{l}
\hspace{0.6mm}	b::=\!\!\\
	\\
	f::=\!
	\end{array} &&\vert \begin{array}{l}
\texttt{\textbf{if}}~e~\texttt{\textbf{then}}~t~\texttt{\textbf{else}}~t \\
\texttt{\textbf{match}}~e~\texttt{\textbf{with}}~|_i~\texttt{C}_i~(x_i^1 \ldots x_i^n)\mapsto t_i \\
\texttt{\textbf{fix}}~\texttt{f}.\lambda x_1 \ldots \lambda x_n.t \\
	\end{array}
	\end{array}\]
	\caption{\tool syntax.}
	\label{Fi:syntax}
\end{figure}

\mypar{Syntax.}
Consider the language shown in \figref{syntax}.
%\footnote{This language  is a variant of the one presented by Osera and Zdancewic \cite{osera2015type}.}
In the language, we distinguish between two kinds of terms:
\textit{elimination terms} (E-terms) and
\textit{introduction terms} (I-terms).
E-terms consist of variable terms, constant values $c$, and application terms.
Condition guards and match scrutinies can only be E-terms.
I-terms are branching terms and function terms.
The key property of I-terms is that if the type of any I-term is known,
the types of its sub-terms are also known (which is not the case for E-terms).

\begin{figure}[t!]
	\[\arraycolsep=1.4pt
	\begin{array}{lrclrcl}
	%	\text{Size, index}&   a,c &&\text{Non-negative integers}\\
	%	\text{Integer expr.}& I::=&&x~|~c~|~I+I~|~I-I~|~I\div I~|~\cdots\\
	\text{Logical expr.}	&\varphi,\phi, \psi::=&& x~|~\texttt{m}(\psi)~|~ \top~|~\bot~|~c~|~\psi~\texttt{mod}~\psi~|~\psi\wedge\psi~|~\psi\vee\psi \\
	&|&& ~\neg\psi~|~\psi=\psi~|~\psi*\psi~|~\psi/\psi~|~\psi+\psi~|~\psi-\psi \\
	\text{Ordinary type}&	B ::=&& \texttt{Bool}~|~\texttt{Int}~|~D \\
	\text{Refinement type} &\tau::=&&\{B~|~\varphi\}~|~x_1\!:\!\tau_1\!\to \ldots \to\!x_n\!:\!\tau_n\!\to\!y:\tau \\
	\text{Annotated type}&	\gamma ::=&&\langle\tau;\alpha\rangle \\
	\text{Recurrence ann.}&	\alpha ::= &&(\generatecolor{~[c_1,\phi_1]_{\texttt{f}}, \ldots ,[c_n,\phi_n]_{\texttt{f}};O(\psi)~}) \\
	\text{Environment}  &\Gamma::=&& \cdot~|~x:\gamma;\Gamma ~|~\varphi;\Gamma~
	|~\texttt{recFun}:=x;\Gamma~|~\texttt{args}:=x_1 \ldots x_n;\Gamma
	\end{array}\]
	\caption{\name types.}
	\label{Fi:types}
\end{figure}

\mypar{Types.}
	Our language of types, presented in \figref{types}, extends the one
of \synquid~\cite{polikarpova2016program} with \emph{recurrence annotations},
which are used to track recurrence relations on functions.
To simplify the presentation, we ignore some of the features of the type
system of \synquid~\cite{polikarpova2016program} that do not affect our algorithm.
In particular, we do not discuss polymorphic types and the enumerating strategy
that ensures that only terminating programs are synthesized.
However, our implementation is built on top of \synquid, and supports
both of those features.

Logical expressions  are built from variables, constants, 
arithmetic operators, and other user-defined logical functions.
Logical expressions in our type system can be used as refinements $\varphi$, size
expressions $\phi$, or bound expressions $\psi$.
Refinements $\varphi$ are logical predicates used to refine ordinary types in refinement
types $\{B~|~\varphi\}$. We usually use a reserved symbol $\emph{v}$ as the free variable in $\varphi$, and let $\emph{v}$ represents the inhabitants, i.e., inhabitants of the type $\{B~|~\varphi\}$ are valuations of $\emph{v}$ that satisfy $\varphi$. For example, the type $\{\texttt{Int}~|~\emph{v}~\texttt{mod}~2=0\}$ represents the even integers.
Size expressions and bound expressions are used in recurrence annotations, and are explained later.

Ordinary types includes primitive types
and user-defined algebraic datatypes $D$. 
Datatype constructors $\texttt{C}$ are functions of type $\tau_1\!\to \ldots \to\!\tau_n\to D$. For example, the datatype  \texttt{List(Int)} has two constructors: $\texttt{Cons}:\texttt{Int}\!\to\!\texttt{List(Int)}\!\to\!\texttt{List(Int)}$, and $\texttt{Nil}:\texttt{List(Int)}$.
Refinement types are ordinary types refined with some predicates $\psi$, or  arrow types.
Note that, unlike \synquid's type system, \tool's type system does not support higher-order functions---i.e., arguments of functions have to be non-arrow types. All occurrences of $\tau_i$ and $\tau$ in  arrow types $x_1\!:\!\tau_1\!\to \ldots \to\!x_n\!:\!\tau_n\!\to\!y:\tau$ have to be ordinary types or refined ordinary types.
We will discuss this limitation in \sectref{RelatedWork}.

We use \texttt{recFun} to denote the name of the function for which we are performing type-checking, and 
\texttt{args} to denote the tuple of arguments to \texttt{recFun}. For example, in the function \texttt{prod} shown in \eqref{prod}, \texttt{recFun=prod} and $\texttt{args=x~y}$
An environment $\Gamma$ is a sequence of variable bindings $x:\gamma$, path conditions $\varphi$, and assignments for variables \texttt{recFun} and \texttt{args}. 

\mypar{Recurrence Annotations.} 
Annotated types are refinement types annotated with recurrence annotations.
A recurrence annotation is a pair $(\generatecolor{[c_1,\phi_1]_{\texttt{f}}, \ldots ,[c_n,\phi_n]_{\texttt{f}};\,O(\psi)})$ consisting of
(1) a set of recursive-call costs of the form $\generatecolor{[c_i,\phi_i]_{\texttt{f}}}$, and 
(2) a resource-usage bound of the form $\generatecolor{O(\psi)}$. 
Intuitively, a recurrence annotation tracks the number $c_i$ of
recursive calls to $\texttt{f}$ of size $\phi_i$ in the first element
$\generatecolor{[c_1,\phi_1]_{\texttt{f}},  \ldots ,[c_n,\phi_n]_{\texttt{f}}}$ of the pair, as well as the asymptotic
resource usage of the \emph{body} of the function (the second element
$\generatecolor{O(\psi)}$).
Using these quantities, we can compute a recurrence relation
describing the resource usage of the function \texttt{recFun}.
For example, the recurrence annotation
$\generatecolor{([1,\emph{u}-1]_{\texttt{f}},[1,\emph{u}-2]_{\texttt{f}};O(1))}$ corresponds to
the recurrence relation $T_{\texttt{f}}(\emph{u})\le T_{\texttt{f}}(\emph{u}-1)+T_{\texttt{f}}(\emph{u}-2)+O(1)$.

A \textit{recursive-call cost} $\generatecolor{[c,\phi]_{\texttt{f}}}$ associated
with a function $\texttt{f}$ denotes that the body of $\texttt{f}$ can
contain up to $c$ recursive calls to subproblems that have sizes up to
the one specified by size expression $\phi$.
A size expression, $\phi$, is a polynomial over a reserved variable
symbol $\emph{u}$ that represents the size of the top-level problem.
In our paper,  a \emph{problem} with respect to a function
$g::x_1\!:\!\tau_1\!\to \ldots \to\!x_n\!:\!\tau_n\!\to\!y\!:\!\tau$ is
a tuple of terms $e_1 \ldots e_n$, to which $g$ can be applied---i.e., $e_i$
has type $\tau_i$ for all $i$ from $1$ to $n$.
For the problems of function $g$, the size of each problem is defined
by a \emph{size function}
\iffalse measure\footnote{Standard measures take only one input and are defined inductively. In our paper, size measures are syntax sugar for logic integer formulas over standard measures. For example, a size measure over two list can be $\lambda l_1.\lambda l_2.\texttt{(len}~l_1)+(\texttt{len}~l_2)$ where $\texttt{len}$ is a standard measure computing length of lists.}
\fi $\texttt{size}_g$---a  user-defined logical function that has type
$\tau_1\!\to \ldots \to\!\tau_n\!\to\!\texttt{Int}$; i.e., it takes a problem
of $g$ as input and outputs a non-negative integer.
In the body of $g$, we say that a recursive-call term $g~e_1 \ldots e_n$
\emph{satisfies} a size expression $\phi$ if for all $x_1$, $\ldots$, $x_n$,
$\texttt{size}_g~\sem{e_1} \ldots \sem{e_n}\le[(\texttt{size}_g~x_1 \ldots x_n)/\emph{u}]\phi$,
where the $x_i$'s are the arguments of $g$ and the $\sem{e_i}$'s are the evaluations of $e_i$ on input $x_1 \ldots x_n$.
(See \sectref{semantics} for the formal definition of $\sem{\cdot}$.)
Note that one annotation can contain multiple recursive-call costs,
which allows the function to make recursive calls to sub-problems with
different sizes.
We often abbreviate $\langle \tau,\generatecolor{(O(1))}\rangle$ as $\tau$ and
omit $\texttt{f}$ in recursive-call costs if it is clear from context.

A resource bound $\generatecolor{O(\psi)}$  of a non-arrow type
specifies the  bound of the resource usage strictly within the
	top-level-function
body. 
A resource bound in a signature of an auxiliary function $f$ specifies the resource
usage of $f$.
\emph{Bound expressions} $\psi$ in $\generatecolor{O(\psi)}$ are of the form $\emph{u}^a\log^b
\emph{u}+c$ where $a$, $b$, and $c$ are all non-negative constants, and $\emph{u}$ represents the size of the top-level problem.
\iffalse
\john{The footnote says for auxiliary functions we assume $b=0$. Do bound expressions apply to non-auxiliary functions?} \footnote{We assume that $b$ is always 0 when it is in the signature types of auxiliary functions, and can be non-zero only when it is in the specification types \john{Whats specification types?. How about "... is in the top-level specification?}.}
\fi

\begin{example}
	In the function \texttt{prod} (\eqref{prod_log}), the recursive-call term \texttt{prod (div2 x) y} satisfies the recursive-call cost $\generatecolor{[1,\floor{\frac{\emph{u}}{2}]}}$, because $\texttt{size}_\texttt{prod}=\lambda z.\lambda w.z$, 
	and
	$$\texttt{size}_\texttt{prod}~\sem{(\texttt{div2}~\texttt{x})}~\sem{\texttt{y}}=~\sem{\texttt{div2}~\texttt{x}}= \floor{\frac{\texttt{x}}{2}}=[(\texttt{size}_\texttt{prod}~\texttt{x}~\texttt{y})/\emph{u}]\floor{\frac{\emph{u}}{2}}.$$
\end{example}

\subsection{Semantics and Cost Model}
\label{Se:semantics}
We introduce the \emph{concrete}-cost semantics of our language here.
The semantics serves two goals: 
(1) it defines the evaluation of terms (i.e., how to obtain values),
which can be used to compute the sizes of problems in application
expressions, and
(2) it defines the resource usages of terms.

Besides the syntax shown in \figref{syntax}, implementations of
auxiliary functions can contain calls to a tick function
$\texttt{tick}(c,t)$, which specifies that $c$ units of a resource are
used, and the overall value is the value of $t$. 
Note that in our synthesis language, we are not actually synthesizing programs with $\texttt{tick}$ functions. We assume that $\texttt{tick}$ functions are only called in the implementations of auxiliary functions.
In the concrete-cost semantics, a configuration $\concreteConfig{t,C}$
consists of a term $t$ and a nonnegative integer $C$ denoting the
resource usage so far.
The evaluation judgment
$\concreteConfig{t,C} \hookrightarrow \concreteConfig{t',C+C_\Delta}$ states
that a term $t$ can be evaluated in one step to a term (or a value) $t'$, with
resource usage $C_\Delta$.
We write $\concreteConfig{t,C} \hookrightarrow^* \concreteConfig{t',C+C_\Delta}$ to indicate
the reduction from $t$ to $t'$ in zero or more steps.
All of the evaluation judgments are standard, and are shown in \sectref{aSemantics}.
Here we show the judgment of the tick function, where resource usage happens.
$${\inferrule*[right=\textsc{Sem-Tick }]
	{ }
	{\concreteConfig{\texttt{tick}(c,t),C} \hookrightarrow \concreteConfig{t,C+c}}
}$$ 
For a term $t$, $\sem{t}$  denotes the evaluation result of $t$, i.e., $\concreteConfig{t,\cdot} \hookrightarrow^*\concreteConfig{\sem{t},\cdot}$.
\begin{example}\label{Exa:double}
Consider the following function that doubles its input.
\[
  \texttt{fix}~\texttt{double}.\lambda \texttt{x}.
       \texttt{\textbf{if}}~\texttt{x = 0}~
               \texttt{\textbf{then}}~\texttt{0}~
               \texttt{\textbf{else}}~\texttt{tick}(\texttt{1,2 + double(x-1)}).
\]
Let $t_{\texttt{body}}$ denote the function body $   \texttt{\textbf{if}}~\texttt{x=0}~
\texttt{\textbf{then}}~\texttt{0}~
\texttt{\textbf{else}}~\texttt{tick}(\texttt{1,2+double(x-1)})$.
The result of evaluating $\texttt{double}$ on input $5$ is $10$, with resource usage $5$.
\begin{eqnarray*}
&&\concreteConfig{(\texttt{fix}~\texttt{double}.\lambda \texttt{x}.t_{\texttt{body}})\texttt{5},0}\\
&\hookrightarrow&\concreteConfig{\texttt{\textbf{if}}~\texttt{5=0}~
	\texttt{\textbf{then}}~\texttt{0}~
	\texttt{\textbf{else}}~\texttt{tick(1,2+double(4))},0}\\
&\hookrightarrow&\concreteConfig{\texttt{\textbf{if}}~\texttt{false}~
	\texttt{\textbf{then}}~\texttt{0}~
	\texttt{\textbf{else}}~\texttt{tick(1,2+ double(4))},0}\\
&\hookrightarrow&\concreteConfig{\texttt{tick(1,2+double(4))},0}
\hookrightarrow\concreteConfig{\texttt{2+double(4)},1}\\
&\hookrightarrow&\concreteConfig{\texttt{2+(}\texttt{fix}~\texttt{double}.\lambda x.t_{\texttt{body}}\texttt{)}4,1}\hookrightarrow^*\concreteConfig{\texttt{4+double(3)},2}\hookrightarrow^*\concreteConfig{\texttt{10+double(0)},5}
\\
&\hookrightarrow&\concreteConfig{\texttt{\textbf{10+(if}}~\texttt{0=0}~
	\texttt{\textbf{then}}~\texttt{0}~
	\texttt{\textbf{else}}~\texttt{tick(1,2+double(0-1)}\texttt{))},5}\\
&\hookrightarrow&\concreteConfig{\texttt{\textbf{10+(if}}~\texttt{true}~
	\texttt{\textbf{then}}~\texttt{0}~
	\texttt{\textbf{else}}~\texttt{tick(1,2+double(0-1)}\texttt{))},0}
\hookrightarrow
\concreteConfig{\texttt{10+0},5}
\end{eqnarray*}
\end{example}

With the standard concrete semantics, the complexity of a function $f$
is characterized by its resource usage when the function is evaluated
on inputs of a given size.
\begin{definition}[Complexity]\label{De:complexity}
  Given a function $\texttt{fix}~f.\lambda \overline{y}.t$ of type $:\tau_1\to\tau_2$,
  with size function $\texttt{size}_f:\tau_1\to\mathbb{N}$, and
  suppose that for any possible input $\overline{x}$, the configuration
  $\concreteConfig{(\texttt{fix}~f.\lambda \overline{y}.t)\overline{x},0}$
  can be reduced to $\concreteConfig{v, C_{\overline{x}}}$ for some value $v$.
  Then, if $T_f:\mathbb{N}\to\mathbb{N}$ is a function such that,
  $
    \textit{for all,}~\emph{u}\ge0,~T_f(\emph{u}) = \sup_{\overline{x}~\textit{s.t.}~\texttt{size}_f(\overline{x}) = \emph{u}}C_{\overline{x}},
 $
  we say that $T_f$ is \textbf{\emph{the complexity function of $f$}}.
\end{definition}
Note that \defref{complexity} assumes that the top-level term
$(\texttt{fix}~f.\lambda \overline{y}.t)\overline{x}$ can be reduced to some value.
Thus, \defref{complexity} only applies to terminating programs. 
\begin{definition}[Big-O notation]\label{De:bigO}
	Given two integer functions $f$ and $g$, we say that $f$ dominates $g$, i.e., $g\in O(f)$, if
	$\exists c,M\ge 0.\  \forall x\ge c.\ g(x)\le Mf(x).$
\end{definition}
In the rest of the paper, we use $T_f$ to denote the complexity function of the function $f$, and we say the complexity of $f$ is \emph{bounded} by a function $g$ if $T_f\in O(g)$.
As an example, the complexity of the \texttt{double} function shown in \exref{double} is $T_{\texttt{double}}(\emph{u}):=\emph{u}$, and hence $T_{\texttt{double}}(\emph{u})\in O(\emph{u})$.

\mypar{Auxiliary functions.}
We allow users to supply signatures for auxiliary functions, instead of implementations.
It is an obligation on users that such signatures be sensible;
in particular, when the user gives the signature
$\langle\tau_1\!\to\!\{B~|~\varphi(\emph{v},\overline{y})\},\generatecolor{O(\psi(\emph{u}))}\rangle$
for auxiliary function $f$, the user asserts that there exists some implementation
$\texttt{fix}~f.\lambda \overline{y}.t$ of $f$, such that: 1)for any input $\overline{x}$, the output of $f$ on $\overline{x}$ satisfies $\varphi$, i.e., $\varphi(\sem{(\texttt{fix}~f.\lambda \overline{y}.t)\overline{x}},\overline{x})$ is valid; and 2)for any input $\overline{x}$, the complexity  of $f$ is bounded by $\psi(\emph{u})$, i.e., $T_f(\emph{u})\in O(\psi(\emph{u}))$.
Signatures always over-approximate their implementations, as
illustrated by the following example.
\begin{example}
The signature
$  \texttt{doubleRelaxed}::\langle \texttt{x:Int}\to\{\texttt{Int}~|~\emph{v} \leq 3*x\},\generatecolor{O(\emph{u}^2)}\rangle$
describes an auxiliary function that computes \textit{no more} than the input times 3, and has quadratic resource usage.
Note that the function $\texttt{double}$ shown in \exref{double} can be an implementation of this signature
because $\sem{\texttt{double}(\overline{x})}=2*x \leq 3*x$, and the complexity function $T_{\texttt{double}}(\emph{u})=\emph{u}$
is in $O(\emph{u}^2)$. 
\end{example}
%The combine $\otimes$ is defined as
%$$\gamma_1\otimes\gamma_2=\generatecolor{([\max(c_1,d_1),\psi_1]_\texttt{f}, \ldots ,[\max(c_n,d_n),\psi_n]_\texttt{f};O(\max(\psi,\phi))}$$

\subsection{Typing Rules}
\label{Se:TypingRules}
The typing rules of \name are inspired by bidirectional type checking
\cite{pierce2000local} and type checking with cost sharing \cite{knoth2019resource}. 
Recall that we use \texttt{recFun} to denote the name of the function for which we are performing type-checking, and 
\texttt{args} to denote the tuple of arguments to \texttt{recFun}.

An environment $\Gamma$ is a sequence of variable bindings of the form
$x:\gamma$, path conditions $\varphi$, and assignments of the form
$x=\varphi$ for \texttt{recFun} and the components of \texttt{args}.
\name's typing rules use three judgments: 1)
$	\Gamma\vdash t\irefine \gamma\text{ states that }t\text{ has type }\gamma$, 2)
$	\Gamma\vdash \gamma_1<:\gamma_2 \text{ states that }\gamma_2\text{ is a subtype of }\gamma_1$, and 3)
	$\Gamma\vdash \gamma\share\gamma_1|\gamma_2\text{ states that }\gamma_1\text{ and }\gamma_2\text{ share the costs in }\gamma $

\mypar{Subtyping.} 
	Subtyping judgments are shown in \figref{subtype} in \appref{aSemantics}. The \textsc{<:-Fun}, \textsc{<:-Sc}, and \textsc{<:-Refl} 
	are standard subtyping rules for refinement types. The remaining rules allow us to compare resource consumption of recurrence annotations. For example, if one branch of some branching term has type $\langle\tau,(\generatecolor{[1,\floor{\frac{\emph{u}}{3}}],O(\psi)})\rangle$, it can be over-approximated by a super type $\langle\tau,(\generatecolor{[1,\floor{\frac{\emph{u}}{2}}],O(\psi)})\rangle$. The idea is that the resource usage of an application calling to a problem of size $\generatecolor{\floor{\frac{\emph{u}}{2}}}$, will be larger than the application calling to a smaller problem of size $\generatecolor{\floor{\frac{\emph{u}}{3}}}$ (assuming all resource usages are monotonic). 
	
	Subtyping rules also allow the type system to compare branches with a different number of recursive calls. For example, base cases of recursive procedures have no recursive calls, and thus have types of the form $\langle \tau, (\generatecolor{[], O(\psi)})\rangle$. 
	With subtyping, these types can be over-approximated by types of the form $\langle \tau, (\generatecolor{[c, \phi], O(\psi)})\rangle$.

\mypar{Cost sharing.} 
When a term has more than one sub-term in the same path, e.g., the condition guard and the \texttt{then} branch are in the same path in an \texttt{ite} term, the recursive-call costs of the term will be shared into its sub-terms. 
The sharing operator $\alpha\share\alpha_1|\alpha_2$ partitions the recursive-call costs of $\alpha$ into $\alpha_1$  and $\alpha_2$---i.e., the sum of the costs in $\alpha_1$ and $\alpha_2$ equals the costs in $\alpha$. Sharing rules are shown in \figref{share}.   The idea is that a single cost $c$ can be shared to two costs $c_1$ and $c_2$ such that their sum is no more than $c$. An annotation can be shared to two parts if every recursive cost $\generatecolor{[c_i,\phi_i]}$ in it can be shared to two parts $\generatecolor{[c_i^{1},\phi_1]}$ and $\generatecolor{[c_i^{2},\phi_2]}$. Finally, annotations can also be shared to more than two parts.
\begin{example}
	There are multiple ways to share the recurrence annotation $\generatecolor{([1,\floor{\frac{\emph{u}}{2}}],[1,\ceil{\frac{\emph{u}}{2}}];O(\emph{u}))}$:
	$$\Gamma\vdash\generatecolor{([1,\floor{\frac{\emph{u}}{2}}],[1,\ceil{\frac{\emph{u}}{2}}];O(\emph{u}))}\share\generatecolor{([1,\floor{\frac{\emph{u}}{2}}],[1,\ceil{\frac{\emph{u}}{2}}];O(\emph{u}))}~|~\generatecolor{([\ ],O(\emph{u}))},$$
	where one annotation contains both recursive-call costs $\generatecolor{[1,\floor{\frac{\emph{u}}{2}}],[1,\floor{\frac{\emph{u}}{2}}]}$; and the other contains no recursive-call cost.
	And
	$$\Gamma\vdash\generatecolor{([1,\floor{\frac{\emph{u}}{2}}],[1,\ceil{\frac{\emph{u}}{2}}];O(\emph{u}))}\share\generatecolor{([1,\floor{\frac{\emph{u}}{2}}];O(\emph{u}))}~|~\generatecolor{([1,\ceil{\frac{\emph{u}}{2}}];O(\emph{u}))},$$
	where each annotation contains one recursive-call cost.
\end{example}

\mypar{Function terms.}
The rule \textsc{T-Abs} shown below is really a rule-schema that is parameterized in terms of an annotation \generatecolor{(A)} for a function body $t$, and a resource bound $\generatecolor{(B)}$ for the function term. If the function body $t$ has some recurrence relation  described by the annotation $\generatecolor{A}$, then the function $\texttt{f}$ will satisfy the resource-usage bound $\generatecolor{B}$.  
Some example patterns are shown in \tableref{pattern}.\footnote{
	The patterns shown in \tableref{pattern} are those we used in the implementation. Patterns capturing other recurrence relations can be added to the type system if needed. }
\[	{\inferrule*[right=\textsc{T-Abs}]
	{\Gamma'=[\texttt{recFun}\gets \texttt{f}][\texttt{args}\gets x_1 \ldots x_n]\Gamma\\ 			
		\gamma_f=\langle {x}_1:\tau_{1}\to \ldots \to x_n:\tau_n\to\tau,\generatecolor{(B)}\rangle\\	\Gamma';x_1\!:\!\langle\tau_1,\generatecolor{O(1)}\rangle; \ldots; x_n\!:\!\langle\tau_n,\generatecolor{O(1)}\rangle;\texttt{f}:\gamma_\texttt{f}\vdash t\irefine\langle\tau,\generatecolor{(A)} \rangle
	}
	{ \Gamma\vdash \texttt{fix}~\texttt{f}.\lambda x_1 \ldots \lambda x_n.t\irefine \langle {x}_1:\tau_{1}\to \ldots \to x_n:\tau_n\to\tau,\generatecolor{(B)}\rangle}
}\]
For example, if the  annotation of the function body is $\generatecolor{([1,\floor{\frac{\emph{u}}{2}}];O(1))}$, then the resource bound in the function type will be $\generatecolor{O(\log \emph{u})}$, i.e., the resource usage of $\texttt{f}$ is bounded by $\generatecolor{O(\log (\texttt{size}_\texttt{f}~x_1 \ldots x_n))}$.

At the same time, the rule stores the name $\texttt{f}$ of the recursive
function into \texttt{recFun}, and its arguments as a tuple into
$\texttt{args}$.
 \begin{table*}[t]
 	\caption{Annotations that can be used  to instantiate the rule \textsc{T-Abs}. }
 	\begin{tabular}{l|l|l|l}
 		&  $Bound~\generatecolor{(B)}$ & Recurrence relation & $Annotation~\generatecolor{(A)}$  \\
 		\hline
 		Master Theorem		& $O(\log\emph{u})$  & $T(\emph{u})\le T(\floor{\frac{\emph{u}}{d}})+O(1),~d\ge2$ & $\generatecolor{([1,\floor{\frac{\emph{u}}{d}}];O(1))}$,~$d\ge2$  \\
 		& $O(\emph{u}\log \emph{u})$ & $T(\emph{u})\le dT(\floor{\frac{\emph{u}}{d}})+O(\emph{u}),~d\ge2$ & $\generatecolor{([d,\floor{\frac{\emph{u}}{d}}];O(\emph{u}))}$,~$d\ge2$ \\
 		Akra–Bazzi 		& $O(\emph{u}\log \emph{u})$ & $T(\emph{u})\le T(\ceil{\frac{\emph{u}}{2}})+T(\floor{\frac{\emph{u}}{2}})+O(\emph{u})$ & $\generatecolor{([1,\ceil{\frac{\emph{u}}{2}}],[1,\floor{\frac{\emph{u}}{2}}];O(\emph{u}))}$\\
 		C-Finite Seq. & $O(\emph{u})$ &  $T(\emph{u})\le T(u-d)+O(1),~d\ge1$ & $\generatecolor{([1,\emph{u}-d];O(1))}$,~$d\ge1$\\
 		& $O(\emph{u}^2)$ &  $T(\emph{u})\le T(\emph{u}-d)+O(\emph{u}),~d\ge1$ & $\generatecolor{([1,\emph{u}-d];O(\emph{u}))}$,~$d\ge1$
 	\end{tabular}
 	\label{Ta:pattern}
 \end{table*}
\begin{example}
We use a function $\texttt{fix~bar}.\lambda x.\texttt{if}~x=1~\texttt{then}~1~\texttt{else}~1+\texttt{bar}(\texttt{div2}~x)$
to illustrate the first pattern in \tableref{pattern}.
The body of $\texttt{bar}$ has the annotated type
$(\generatecolor{[1,\floor{\frac{\emph{u}}{2}}];O(1)})$
because (i) there exists only one recursive call to a sub-problem whose
size is half of the top-level problem size $\emph{u}$, and (ii) the
resource usage inside the body is constant (with the assumption that
all auxiliary functions have constant resource usage).
This type appears in row 1, column 4 of \tableref{pattern}.
Consequently, the recurrence relation of $\texttt{bar}$ is $T(\emph{u})\le
T(\floor{\frac{\emph{u}}{2}})+O(1)$ (row 1, column 3),
where $T(\emph{u})$
is the resource usage of $\texttt{bar}$ on problems with size $\emph{u}$.
Finally, according to the Master Theorem, the resource usage of
$\texttt{bar}$ is bounded by $O(\log\emph{u})$ (row 1, column 2).
\end{example}

\mypar{Branching terms.}
In rule \textsc{T-If}, the condition has type \texttt{Bool} with refinement $\varphi_e$.  Two branches have different types---the \texttt{then} branch follows the path condition $\varphi_e$, and the refinement $\varphi$ of the branch term, while the \texttt{else} branch follows the path condition $\neg\varphi_e$.
By having both branches share the same recurrence annotation,
\textsc{T-If} can introduce some imprecision.
In particular, if the branches belong to different complexity classes,
the annotation of the conditional term will be the upper bound of both
branches.
\[{\inferrule*[right=\textsc{T-If}]
	{\Gamma\vdash \alpha\share\alpha_1|\alpha_2\quad  \Gamma \vdash e\eguess \langle\{\texttt{Bool}~|~\varphi_e\},\generatecolor{\alpha_1}\rangle\\
		\Gamma,\varphi_e \vdash t_1\irefine \langle\{B~|~\varphi\},\generatecolor{\alpha_2}\rangle \quad\Gamma,\neg\varphi_e\vdash t_2\irefine \langle\{B~|~\varphi\},\generatecolor{\alpha_2}\rangle
	}
	{ \Gamma\vdash \texttt{\textbf{if}}~e~\texttt{\textbf{then}}~t_1~\texttt{\textbf{else}}~t_2\irefine\langle\{B~|~\varphi\},\generatecolor{\alpha}\rangle}}\]
\vspace{-4mm}

The  rule \textsc{T-Match} (\sectref{a3Typing}) is slightly different:
(1) there can be more than two branches,
(2) all branches have the same type $\langle \tau,\alpha_2\rangle$, and
(3) variables in each case $\texttt{C}_i~(x_i^1 \ldots x_i^n)$ are introduced in the corresponding branch.

\mypar{E-terms.}
The typing rules for E-terms are shown in \figref{rulesE}.
The two rules for application terms are the key rules of our type system.
Let us first look at the \textsc{E-RecApp} rule for recursive-call terms. Recall that the recursive-call annotation tracks the number of recursive calls and the sizes of sub-problems.
If the term $\texttt{f}~e_1 \ldots e_n$ is a recursive call---i.e., $\Gamma(\texttt{recFun})=\texttt{f}$---the number of recursive calls in one of the recursive-call costs will increase by one---i.e., $\generatecolor{[c_k,\phi_k]}$ in the premise becomes $\generatecolor{[c_k+1,\phi_k]}$ in the conclusion.
Also, we want to make sure that the size of the subproblem this application term is called on  satisfies the size expression $\generatecolor{\phi_k}$. If  each callee term is refined by the predicate $\varphi_i$, i.e., $\Gamma\vdash e_i\eguess \langle \{B_i~|~\varphi_i\},\generatecolor{\alpha_i}\rangle$ , then the fact that the size of the problem $e_1 \ldots e_n$ satisfies $\phi_k$ can be implied by the validity of the predicate $\bigwedge_{i=1}^m[y_i/\emph{v}]\varphi_i\!\Rightarrow\! (\texttt{size}~y_1 \ldots y_m\le[\texttt{size}~\Gamma(\texttt{args})/\emph{u}]\phi_k)$. 
We introduce validity checking, written $\Gamma\models\varphi$ , to state that a predicate expression $\varphi$ is always true under any instance of the environment $\Gamma$. 

\begin{figure*}[t!]
\begin{mathpar}
 			
 			\boxed{\Gamma\vdash e\eguess\gamma}\hspace{10mm}{\inferrule*[right=\textsc{E-SubType}]
 				{\Gamma\vdash e\eguess\gamma'\qquad \Gamma\vdash\gamma'<:\gamma}
 				{ \Gamma\vdash e\eguess\gamma}
 			}
 		\hspace{10mm}
 		{\inferrule*[right=\textsc{E-Var}]
 			{\Gamma(x)=\gamma}
 			{ \Gamma\vdash x\eguess\gamma}
 		}\\\\
 	{\inferrule*[right=\textsc{E-App}]
 		{\Gamma\vdash \texttt{g}: \langle x_1\!:\!\tau_1\!\to \ldots \to\!
 			x_m\!:\!\tau_m\to\{B~|~\varphi\}, \generatecolor{(O(\psi_{\texttt{g}}))}\rangle\\ \Gamma(\texttt{recFun})\ne\texttt{g}\\ \Gamma\vdash \generatecolor{([c_1,\phi_1], \ldots , \ldots ,[c_n,\phi_n];O(\psi))}\share \alpha_1| \ldots |\alpha_m \\
 			\forall 1\le i\le m\qquad	\Gamma\vdash e_i\eguess \langle \{B_i~|~\varphi_i\}, \generatecolor{\alpha_i}\rangle\qquad \Gamma\vdash \{B_i~|~\varphi_i\}<:\tau_i\\
 			\Gamma\models\bigwedge_{i=1}^m[y_i/\emph{v}]\varphi_i\!\Rightarrow\! \left([\texttt{size}_{\texttt{g}}~y_1 \ldots y_m/\emph{u}]\psi_{\texttt{g}}\in O([\texttt{size}~\Gamma(\texttt{args})/\emph{u}]\psi)\right)\\
 			\tau = \{B~|~[z_i/x_i]\varphi \land \bigwedge_{i=1} [z_i/v]\varphi_i\}~z_i\notin FV(\varphi), z_i\notin FV(\varphi_i)
 		}
 		{ \Gamma\vdash \texttt{g}e_1 \ldots e_m\eguess\langle\tau,\generatecolor{([c_1,\phi_1], \ldots ,[c_n,\phi_n];O(\psi))}\rangle}
 	}\\\\
 {\inferrule*[right=\textsc{E-RecApp}]
 	{\Gamma\vdash \texttt{f}: \langle x_1\!:\!\tau_1\!\to \ldots \to\!
 		x_m\!:\!\tau_m\to\{B~|~\varphi\}, \generatecolor{\alpha}\rangle\qquad \Gamma(\texttt{recFun})=\texttt{f}\\ \Gamma\vdash \generatecolor{([c_1,\phi_1], \ldots ,[c_k,\phi_k], \ldots ,[c_n,\phi_n];O(\psi))}\share \alpha_1| \ldots |\alpha_m \\
 		\forall 1\le i\le m\qquad	\Gamma\vdash e_i\eguess \langle \{B_i~|~\varphi_i\}, \generatecolor{\alpha_i}\rangle\qquad \Gamma\vdash \{B_i~|~\varphi_i\}<:\tau_i\\
 		\Gamma\models\bigwedge_{i=1}^m[y_i/\emph{v}]\varphi_i\!\Rightarrow\! (\texttt{size}~y_1 \ldots y_m\le[\texttt{size}~\Gamma(\texttt{args})/\emph{u}]\phi_k)\\
 		\tau = \{B~|~[z_i/x_i]\varphi \land \bigwedge_{i=1} [z_i/v]\varphi_i\}~z_i\notin FV(\varphi), z_i\notin FV(\varphi_i)
 	}
 	{ \Gamma\vdash \texttt{f}e_1 \ldots e_m\eguess\langle\tau,\generatecolor{([c_1,\phi_1], \ldots [c_k+1,\phi_k],\ldots,[c_n,\phi_n];O(\psi))}\rangle}
 }
 	\end{mathpar}
 \vspace{-5mm}
 	\caption{Typing rules of E-terms}
 	\label{Fi:rulesE}
 \end{figure*}

\begin{example}
Recall \eqref{prod_log}. According to the rule \textsc{T-RecApp}, the recursive call $\texttt{prod (div2 x) y}$ has type $\langle\{\texttt{Int}~|~\emph{v}=\floor{\frac{\texttt{x}}{2}}*\texttt{y}\},\generatecolor{([1,\frac{\emph{u}}{2}]);O(1)}\rangle$.
Note that the first argument $\texttt{(div2 x)}$ has type $\{\text{Int}~|~\emph{v}=\floor{\frac{\texttt{x}}{2}}\}$,
the second argument $\texttt{y}$ has type $\{\text{Int}~|~\emph{v}=\texttt{y}\}$,
the size function is $\texttt{size}_\texttt{prod}=\lambda z.\lambda w.z$,
and the arguments in the context are $\Gamma(\texttt{args})=\texttt{x}~\texttt{y}$.
Therefore, the following predicate is valid:
\begin{eqnarray*}
&&	[y_1/\emph{v}](\emph{v}=\floor{\frac{x}{2}})\wedge [y_2/\emph{v}](\emph{v}=y)\!\Rightarrow\!\texttt{size}_\texttt{prod}~y_1~y_2=[\texttt{size}_{\texttt{prod}}~\Gamma(\texttt{args}/\emph{u})]\floor{\frac{\emph{u}}{2}}\\
\Leftrightarrow&&(y_1=\floor{\frac{x}{2}})\wedge (y_2=y)\!\Rightarrow\! y_1=\floor{\frac{x}{2}}.
\end{eqnarray*}
\end{example}

The rule \textsc{E-App} states that callees have types $\tau_i$, and the resource usage does not exceed the bound $\generatecolor{O(\psi)}$ in the
annotation.
Similar to the \textsc{E-RecApp} rule, the size of the problem $\texttt{g}$ calls to is 
$[\texttt{size}_{\texttt{g}}~y_1 \ldots y_m/\emph{u}]$ with the premise $\bigwedge_{i=1}^m[y_i/\emph{v}]\varphi_i$.
The validation checking $\bigwedge_{i=1}^m[y_i/\emph{v}]\varphi_i\!\Rightarrow\! \left([\texttt{size}_{\texttt{g}}~y_1 \ldots y_m/\emph{u}]\psi_{\texttt{g}}\in O([\texttt{size}~\Gamma(\texttt{args})/\emph{u}]\psi)\right)$ in the rule states that for any instance of $\Gamma$, the size of the problem in the application term is in the big-O class $O([\texttt{size}~\Gamma(\texttt{args})/\emph{u}]\psi)$. Note that the membership of big-O classes can be encoded as an $\exists\forall$ query. The query is non-linear, and hence undecidable in general.
However, we observed in our experiments that for many benchmarks the query stays linear.
Furthermore, even when the query is non-linear,
existing SMT solvers are capable of handling many such checks in practice.

\subsection{Soundness}
We assume that the resource-usage function $\psi$ and the complexities $T$ of each function are all nonnegative and monotonic 
	integer functions---both the input and the output are integers.
We show soundness of the type system with respect to the resource model. The soundness theorem states that if we derive a bound $O(\psi)$ for a function $\texttt{f}$, then the complexity of $\texttt{f}$ is bounded by $\psi$. 

\begin{theorem}[Soundness of type checking]
	Given a function $\texttt{fix }f.\lambda x_1 \ldots \lambda x_n.t$ and an environment $\Gamma$, if $\Gamma\vdash \texttt{fix }f.\lambda x_1 \ldots \lambda x_n.t:: \langle\tau,\generatecolor{O(\psi)} \rangle$, then the complexity of $f$ is bounded by $\psi$.
\end{theorem}

 Our type system is incomplete with respect to resource usage. That is, there are functions in our programming language that are actually in a complexity class $O(p(x))$, but cannot be typed in our type system.
The main reason why our type system is incomplete is that it ignores condition guards when building recurrence relations, and over-approximates \texttt{if}-\texttt{then}-\texttt{else} terms by choosing
the largest complexity among all the paths including even unreachable ones.

%%% Local Variables: 
%%% mode: latex
%%% TeX-master: "main_cav.tex"
%%% End: 

\section{The \name Synthesis Algorithm}
\label{Se:synthesisAlgorithm}

In this section, we present the \name synthesis algorithm, which uses annotated types to guide the search of terms of given types. 

\subsection{Overview of the Synthesis Algorithm}

The algorithm takes as input
a goal type $\texttt{f}:\langle \tau,\generatecolor{O(\psi)}\rangle$,
an environment $\Gamma$ that includes type information of auxiliary functions,
and the \texttt{size} functions for \texttt{f} and all auxiliary functions.
The goal is to find a function term of type $\langle \tau,\generatecolor{O(\psi)}\rangle$.

The algorithm uses the rules of the \name type system to decompose goal types into
sub-goals, and then applies itself recursively on the sub-goals to
synthesize sub-terms.
Concretely, given a goal $\gamma$, the algorithm tries all the rules shown in \figrefs{rulesI}{rulesE},
where the type in the conclusion matches $\gamma$, to construct sub-goals:
for each sub-term $t$ in the conclusion, there must be a judgment $\Gamma\vdash t\irefine \gamma'$ in the premise;
thus, we construct the sub-goal $\gamma'$---the desired type of $t$.
For each I-term rule (\figref{rulesI}), the type of each sub-term is always known,
and thus a fixed set of sub-goals is generated.
For each E-term rule (\figref{rulesE}), the algorithm enumerates
E-terms up to a certain depth (the depth can be given as a parameter or 
it can automatically increase throughout the search).
If the algorithm fails to solve some sub-goal using some E-term rule,
it backtracks to an earlier choice point, and tries another rule.

Because the top-level goal is always a function type, the algorithm
always starts by applying the rule $\textsc{T-Abs}$, which matches the
resource bound $\generatecolor{O(\psi)}$ using \tableref{pattern} to
infer a possible recurrence annotation for the type of the function
body.
Also $\textsc{T-Abs}$ constructs a sub-goal type for the function
body. In the rest of this section, we assume that goals are not
function types.

\begin{changebar} \begin{algorithm}[ht]\DontPrintSemicolon
		{\small
			\SetKwInOut{Input}{Input}
			\SetKwInOut{Output}{Output}
			\SetKwInOut{Function}{Function}
			\Input{Context $\Gamma$, goal type $\gamma=\langle\{B~|~\varphi\},\generatecolor{\alpha}\rangle$, depth bound $d$}
			\For{$t\gets\textsc{EnumerateE}(\Gamma,d,B)$}{
				\lIf {$\textsc{CheckE}(t,\Gamma,\gamma)$}{\Return $t$}
			}
			\Return $\bot$
		}
		\caption{\textsc{GenerateE}$(\Gamma,\gamma,d)$}
		\label{Alg:enumerateE}
	\end{algorithm}
\mypar{Synthesizing E-Terms.} 
The algorithm for synthesizing E-terms is shown in \algref{enumerateE}.
It enumerates each E-term $t$---with depth up to $d$---that satisfies the base type $B$ in the goal $\gamma:=\langle\{B~|~\varphi\},\generatecolor{([c_1,\phi_1]..[c_n,\phi_n];O(\psi))}\rangle$ from the context $\Gamma$.
For each such E-term $t$, the algorithm checks whether $t$ satisfies the goal type with a subroutine \textsc{CheckE}, which operates as follows.

When $t$ is a variable term, \textsc{CheckE} checks the refined type of $t$ against the goal.
When $t$ is an application term, \textsc{CheckE} first checks if the total number of recursive calls in the term $t$ exceeds the bound $\sum_ic_i$, and if it does, the term $t$ is rejected.
Otherwise, \textsc{CheckE} checks the sizes of sub-problems of recursive calls in $t$.
Formally, to check if a recursive application term $f(t_1,..,t_m)$ is consistent with some $\generatecolor{[c_k,\phi_k]}$, the algorithm queries the validity of the following predicate
$$(\bigwedge_{i=1}^{m}[y_i/\emph{v}]\varphi_i\!\Rightarrow\! (\texttt{size}_f(y_1~..~y_{m})=[\texttt{size}_f(\Gamma(\texttt{args}))/\emph{v}]\phi_k)),$$
where the $y_i$'s are  fresh variables, and the $\varphi_i$'s are the refinements of terms $t_i$'s. 
If the sizes of sub-problems are not consistent with the
recursive-call costs $\generatecolor{[c_1,\phi_1]..[c_n,\phi_n]}$, the
term $t$ is rejected. 
Note that one recursive call can possibly satisfy more than one $\generatecolor{[c_k,\phi_k]}$.
The algorithm enumerates all possible matches.
Finally, \textsc{CheckE} checks the refined type of $t$ against the goal.

Checking the validity of auxiliary application terms is similar. 
\textsc{CheckE} needs to establish that the following predicate holds, which asserts that the resource usage of an auxiliary function does not exceed the bound $O(\psi)$.
\[
  \bigwedge_{i=1}^{m}[y_i/\emph{v}]\varphi_i\!\Rightarrow\! \left([\texttt{size}_{\texttt{g}}~y_1..y_m/\emph{v}]\psi_{\texttt{g}}\in O([\texttt{size}~\Gamma(\texttt{args})/\emph{v}]\psi)\right).
\]
Recall that the above query is undecidable in general, and is checked with best effort by an SMT solver in \name. 

\mypar{Synthesizing I-Terms.}
\algref{enumerateI} shows the algorithm for synthesizing I-Terms.
\textsc{GenerateI} first tries to synthesize an E-term for the goal $\gamma$ (\lineref{e-term}).

If there is no E-term that satisfies the goal, and the match bound $m$ is greater than 0, \textsc{GenerateI} chooses to apply the rule $\text{T-Match}$ (\lineseqref{matchStart}{matchEnd}). First, it enumerates candidate scrutinees $s$, which are E-terms of some data type. Then it generates \texttt{match} $patterns$ according to the type of $s$ (line (3)), updates the goal with a new recursive-call cost (line (4)), and generates case terms $t_i$ for each pattern $pattern[i]$ (lines (5)-(7)). The subroutine \textsc{UpdateCost} is used to subtract the recursive-call cost usage from the cost in $\gamma$. Finally, if all case terms are found, the algorithm constructs the corresponding \texttt{match}-term and returns it.

If there is no \texttt{match}-term  satisfying the goal, \textsc{GenerateI} applies the rule $\textsc{T-If}$  to synthesize a term of the form $\texttt{\textbf{if}}~cond~\texttt{\textbf{then}}~t_T~\texttt{\textbf{else}}~t_F$, and performs three steps to construct sub-goals for sub-terms $cond$, $t_T$, and $t_F$:
(1) it enumerates the condition guard $cond$ (line (10)) of  type \texttt{bool}; 
(2) it updates the cost in the goal $\gamma$ (line (11)); and
(3) it propagates sub-goals to the two branches $t_T$ and $t_F$ with $cond$ and $\neg cond$ as the path condition (lines (12) and (13)), respectively. Finally, if both $t_T$ and $t_F$ are found, the algorithm constructs the corresponding \texttt{if}-term and returns it as a solution (line (14)).

	\begin{algorithm}[tb]\DontPrintSemicolon
{\small
	\SetKwInOut{Input}{Input}
	\SetKwInOut{Output}{Output}
	\SetKwInOut{Function}{Function}
	\Input{Context $\Gamma$, goal type $\gamma$, depth bound $d$, match bound $m$}
	\lIf {$t\gets \textsc{GenerateE}(\Gamma,\gamma,d)$} {
		\Return t \label{Li:e-term}
	}
	\lIf {$m > 0$\label{Li:matchStart}} {\For{$s\gets\textsc{EnumerateE}(\Gamma,d,dataType)$}{
			$patterns\gets \textsc{GeneratePatterns}(\Gamma,\text{TypeOf}(s))$\;
			$\gamma' \gets \textsc{UpdateCost}(s,\gamma)$\;
			\For{$i \in [1,\textsc{Size}(patterns)]$}{$t_i\gets \textsc{GenerateI}(\textsc{UpdateContext}(\Gamma,s==patterns[i]),\gamma',d,m-1)$\;
				\lIf{$t_i==\bot$}{\Return $\bot$}}
			\Return \texttt{Match}~s~\texttt{with}~$|_{i}$~$patterns[i]\to t_i$\label{Li:matchEnd}
	}}
	
	\For{$cond\gets \textsc{EnumerateE}(\Gamma,d,\texttt{Bool})$}{$\gamma' \gets \textsc{UpdateCost}(s,\gamma)$\;
		$t_T\gets \textsc{GenerateI}(\textsc{UpdateContext}(\Gamma,cond),\gamma',d,m)$\;
		$t_F\gets \textsc{GenerateI}(\textsc{UpdateContext}(\Gamma,\neg cond),\gamma',d,m)$\;
		\lIf{$t_T\ne \bot\wedge t_F\ne \bot$}{\Return \texttt{If}~$cond$~\texttt{then}~$t_T$~\texttt{else}~$t_F$}}
	\Return $\bot$
}
\caption{\textsc{GenerateI}$(\Gamma,\gamma,d,m)$.}
\label{Alg:enumerateI}
\end{algorithm}

\mypar{Optimization.}
\algref{enumerateI} discussed above is based on bidirectional type-guided synthesis with liquid types (\synquid \cite{polikarpova2016program}). Therefore, \emph{liquid abduction} and \emph{match abduction}, two optimizations used in \synquid, can also be used in \name. These two techniques allow  one to synthesize the branches of \texttt{if}- and \texttt{match}-terms, and then use  logical abduction to infer the weakest assumption under which the branch fulfills the goal type.
\end{changebar}

\begin{example}
\label{Exa:synthesis}

\begin{figure*}[t!]
	\[
	\renewcommand{\arraycolsep}{0pt}% Remove separation between array columns
	\begin{array}{ll}
	&\texttt{prod=??}_1\texttt{:}\langle\texttt{x:}\{\texttt{Int}~|~\emph{v}\ge 0\}\!\to\!\texttt{y:}\{\texttt{Int}~|~\emph{v}\ge 0\}\!\to\!\{\texttt{Int}~|~\emph{v}=\texttt{x}*\texttt{y}\},\generatecolor{(O(\log \emph{u}))}\rangle\\
	\hdotsfor{2} \\
	\xrightarrow{\textsc{T-Abs}}&\texttt{prod=}\lambda\texttt{x}.\lambda\texttt{y}.\texttt{??}_2\texttt{:}\langle\{\texttt{Int}~|~\emph{v}=\texttt{x}*\texttt{y},\generatecolor{([1,\floor{\frac{\emph{u}}{2}}];O(1))}\rangle\\
	\hdotsfor{2} \\
	\xrightarrow{\textsc{T-If}}&\texttt{prod=}\lambda\texttt{x}.\lambda\texttt{y}.\texttt{\textbf{if}}~\texttt{??}_3   \quad\xleftarrow{\textsc{E-App}}\texttt{x==}0\texttt{:}\langle\{\texttt{Bool}~|~\texttt{x}=0\},\generatecolor{([0,\floor{\frac{\emph{u}}{2}}];O(1))}\rangle\\
	&\hspace{5mm}\texttt{\textbf{then}}~\texttt{??}_4\texttt{:}\langle\{\texttt{Int}~|~\emph{v}=\texttt{x}*\texttt{y}\wedge\texttt{x}=0,\generatecolor{([1,\floor{\frac{\emph{u}}{2}}];O(1))}\rangle\\
	&\hspace{18mm}\xleftarrow{\textsc{E-App}}\texttt{x}\texttt{:}\langle\{\texttt{Int}~|~\emph{v}=0\},\generatecolor{([0,\floor{\frac{\emph{u}}{2}}];O(1))}\rangle\\
	&\hspace{5mm}\texttt{\textbf{else}}~\texttt{??}_5\texttt{:}\langle\{\texttt{Int}~|~\emph{v}=\texttt{x}*\texttt{y}\wedge\texttt{x}>0,\generatecolor{([1,\floor{\frac{\emph{u}}{2}}];O(1))}\rangle\\
	\hdotsfor{2} \\
	\texttt{??}_5\xrightarrow{\textsc{T-If}}~&\texttt{\textbf{if}}~\texttt{??}_6\quad\xleftarrow{\textsc{E-App}}\texttt{even x}\texttt{:}\langle\{\texttt{Bool}~|~\texttt{x mod }2=0\},\generatecolor{([0,\floor{\frac{\emph{u}}{2}}];O(1))}\rangle\\
	&\hspace{5mm}\texttt{\textbf{then}}~\texttt{??}_7\texttt{:}\langle\{\texttt{Int}~|~\emph{v}=\texttt{x}*\texttt{y}\wedge\texttt{x mod }2=0,\generatecolor{([1,\floor{\frac{\emph{u}}{2}}];O(1))}\rangle\\
	%	&\hspace{20mm}\xleftarrow{\textsc{E-App}}\texttt{double ??}_8\texttt{:}\langle\{\texttt{Int}~|~\emph{v}=\texttt{x}*\texttt{y}/2\},\generatecolor{([1,\frac{\emph{v}}{2}];O(1))}\rangle\\
	&\hspace{18mm}\xleftarrow{\textsc{E-App}}\texttt{double (prod (div2 x) y)}\\
	&\hspace{5mm}\texttt{\textbf{else}}~\texttt{??}_9\texttt{:}\langle\{\texttt{Int}~|~\emph{v}=\texttt{x}*\texttt{y}\wedge\texttt{x mod 2}=1,\generatecolor{([1,\floor{\frac{\emph{u}}{2}}];O(1))}\rangle\\
	%	&\hspace{20mm}\xleftarrow{\textsc{E-App}}\texttt{plus y ??}_{10}\texttt{:}\langle\{\texttt{Int}~|~\emph{v}=(\texttt{x}-1)*\texttt{y}\wedge\texttt{x mod 2}=1\},\generatecolor{([1,\frac{\emph{v}}{2}];O(1))}\rangle\\
	&\hspace{18mm}\xleftarrow{\textsc{E-App}}\texttt{plus y (double (prod (div2 x) y))}
	\end{array}\]
	\vspace{-5mm}
	\caption{Trace of the synthesis of an $O(\log x)$ implementation of \texttt{prod}.}
	\label{Fi:prod_full}
	\vspace{-3mm}
\end{figure*}
We illustrate in \figref{prod_full} how the algorithm synthesizes the $O(\log x)$ implementation of $\texttt{prod}$
presented in \eqref{prod_log}. We omit the type contexts in the example.
We will use ``$\texttt{??}$'' to denote intermediate terms being synthesized (i.e., holes in the program).
At the beginning, the type of $\texttt{??}_1$ (i.e., the term we are synthesizing) is an arrow type with resource bound $\generatecolor{O(\log\emph{u})}$ specified by the input goal. 
In this example, \name applies to the arrow type the rule
\textsc{T-Abs}, parameterized according to the first rule in
\tableref{pattern}.
This step produces the sub-problem of synthesizing the function body
$\texttt{??}_2$, whose annotation is
$\generatecolor{([1,\floor{\frac{\emph{u}}{2}}];O(1))}$---which means
that $\texttt{??}_2$ should contain at most one recursive call to
sub-problems with size $\floor{\frac{\emph{u}}{2}}$.

Next, \name chooses to fill $\texttt{??}_2$ with an \texttt{if-then-else} term (by applying the \textsc{T-If} rules) with three sub-problems: the condition guard $\texttt{??}_3$, the \texttt{then} branch $\texttt{??}_4$ and the \texttt{else} branch $\texttt{??}_5$. Note that here we share the number of recursive calls $\generatecolor{[1,\frac{\emph{u}}{2}]}$ 
as follows:  \generatecolor{0} recursive calls in the condition guard,
and \generatecolor{1} in the \texttt{then} branch and the \texttt{else} branch.
The left arrow \texttt{E-App} shows how \name enumerates terms and checks them against the goal types of sub-problems.
For example, to fill $\texttt{??}_4$, \name enumerates terms of type
$\langle\{\texttt{Int}~|~\emph{v}=\texttt{x}*\texttt{y}\wedge\texttt{x}=0,\generatecolor{([1,\frac{\emph{u}}{2}];O(1))}\rangle$,
which are restricted to contain at most one recursive call to \texttt{prod}. 
In \figref{prod_full}, \name has picked the term $\texttt{x}$ to fill $\texttt{??}_4$.
The refinement type of the variable term $\texttt{x}$ is $\{\texttt{Int}~|~\emph{v}=\texttt{x}\wedge\texttt{x}=0\}$ where $\texttt{x}=0$ is the path condition. To check that $\texttt{x}$ also satisfies the type of $\texttt{??}_4$, the algorithm needs to apply rule \textsc{E-SubType}, and check that, for any $\emph{v}$ and $\texttt{x}$, $\emph{v}=\texttt{x}\wedge\texttt{x}=0$ implies $\emph{v}=\texttt{x}*\texttt{y}\wedge\texttt{x}=0$, and $\generatecolor{[0, \floor{\frac{\emph{u}}{2}}]}$ is approximated by $\generatecolor{[1,  \floor{\frac{\emph{u}}{2}}]}$.

After applying another \textsc{T-If} rule for $\texttt{??}_5$, \name
produces three new sub-problems $\texttt{??}_6$, $\texttt{??}_7$, and
$\texttt{??}_8$.
When enumerating terms to fill $\texttt{??}_7$, \name finds an
application term \texttt{double (prod (div2 x) y)} that satisfies the
goal
$\langle\{\texttt{Int}~|~\emph{v}=\texttt{x}*\texttt{y}\wedge\texttt{x mod }2=0,\generatecolor{([1,\floor{\frac{\emph{u}}{2}}];O(1))}\rangle$.
To check that the size of the problem in the recursive call
\texttt{prod (div2 x) y} satisfies the recursive-call cost
$\generatecolor{[1,\floor{\frac{\emph{u}}{2}}]}$, the type system first checks
the refinement of the callee.
The refinement of the first argument \texttt{(div2 x)} is $\varphi_1:=\emph{v}=\floor{\frac{\texttt{x}}{2}}$.
The refinement of the second argument \texttt{y} is $\varphi_2:=\emph{v}=\texttt{y}$.
Consequently, the size of the sub-problem $\texttt{prod (div2 x) y}$ satisfies
$[1,\floor{\frac{\emph{u}}{2}}]$ because $[z/\emph{v}]\varphi_1\wedge[w/\emph{v}]\varphi_2\implies\texttt{size}~z~w=[(\texttt{size}~\texttt{x}~\texttt{y})/\emph{v}]\floor{\frac{\emph{u}}{2}}$,
which can be simplified to $z=\floor{\frac{\texttt{x}}{2}} \wedge w = \texttt{y} \implies z=\floor{\frac{\texttt{x}}{2}}$.
(Recall that the size function for \texttt{prod} is $\texttt{size}:=\lambda z.\lambda w.z$.) 
\end{example}

The algorithm is sound because it only enumerates well-typed terms.

\begin{theorem}[Soundness of the synthesis algorithm]
	Given a goal type $\langle\tau,\generatecolor{O(\psi)}\rangle$ and an environment $\Gamma$, if a term $\texttt{fix }f.\lambda x_1..\lambda x_n.t$ is synthesized by \name, then the complexity of $f$ is bounded by $\psi$.
\end{theorem}

\begin{changebar}
	Although \name does not support higher-order functions in general, it can solve restricted but practical problems with higher-order functions. This extension is discussed in \sectref{hof}.
\end{changebar}

%%% Local Variables: 
%%% mode: latex
%%% TeX-master: "main_cav.tex"
%%% End: 

\section{Extensions to the \name Type System}
\label{Se:extension}
In this section, we introduce two extensions to the \name type system.

\subsubsubsection{Recurrence Relations with Correlated Sizes.} 
The type system shown in \sectref{typeSystem} only tracks sub-problems with independent sizes.
%, i.e., the size expressions $\phi_1$ and $\phi_2$ in a recursive-call cost $\generatecolor{[c_1,\phi_1],[c_2,\phi_2]}$ cannot be dependent on each other. 
For example, consider the recurrence relation $T(\emph{u})=T(l)+T(r)+O(1)$, where
the variables $l$ and $r$ are correlated by the constraint $l+r<\emph{u}$.
This relation is needed to reason about programs that manipulate binary trees or binary heaps,
where $l$ and $r$ represent the sizes of the two children.  
To support such a recurrence relation, we extend \tool's type system with recursive-call costs 
of the form $\generatecolor{[1,l],[1,\emph{u}-1-l]}$, where $l$ is a free variable. 
When correlated recurrence relations are present, the synthesis algorithm will:
(1) match the first enumerated recursive-call term to $\generatecolor{[1,l]}$, and instantiate the size $l$ with $s$, where $s$ is the size of the recursive-call term ($s$ should be smaller than the size $\emph{u}$ of the top-level function); and
(2) use the size $s$ of the recursive-call term computed in step  1 to constrain the algorithm to
enumerate only recursive-call terms of sizes at most $\emph{u}-1-s$.
	
\subsubsubsection{Synthesis of Auxiliary Functions.}
Most of the existing type-directed approaches require the input to the problem to
contain all needed auxiliary functions.
With \tool, some of the auxiliary functions needed
to solve synthesis problems with resource annotations can be synthesized
automatically.

For example, consider the problem \texttt{prod} described in
\sectref{Overview}.
In this problem, we observe that one of the provided auxiliary
functions, $\texttt{div2}$, strongly resembles one of the elements of
the recurrence relation, $T(\emph{u})\le T(\floor{\frac{\emph{u}}{2}})+O(1)$,
needed to synthesize a program with the desired resource usage.
In particular, we know that one needs an auxiliary function that
can take an input of size $u$ and produce an output of size $\floor{\frac{\emph{u}}{2}}$.
In this example, the required auxiliary function $\texttt{div2}$
merely needs to divide the input by $2$ (and round down), but in certain
cases it might need a more precise refinement type than merely
changing the size of the input.
For example, the auxiliary function split used by merge sort needs to
split the input list $\texttt{xs}$ into two lists $\texttt{v1}$ and $\texttt{v2}$
that are half the length of the input \textit{and}
such that 
$\texttt{elems}(\texttt{v1}) \uplus \texttt{elems}(\texttt{v2})=\texttt{elems}(\texttt{xs})$.
However, all we know from the refinement is that the output lists must be half the length of the original list.

Although we do not know what this auxiliary function should do exactly,
we can use the size constraint appearing in the recurrence relation to define part of the refinement type we want the auxiliary 
function to satisfy.
\tool builds on this idea and incorporates an (optionally enabled) algorithm, \nameAuxAlgo, 
that while trying to synthesize a solution to the top-level synthesis problem
also tries in parallel to synthesize auxiliary functions that can create sub-problems with the size constraints needed in the recurrence relation.
To address the problem mentioned above---i.e., that we do not know the exact refinement type the auxiliary function
should satisfy---\nameAuxAlgo enumerates \emph{auxiliary refinements}, which are possible
specifications that the auxiliary function \texttt{aux} we are trying to synthesize might satisfy.

%%% Local Variables: 
%%% mode: latex
%%% TeX-master: "main_cav.tex"
%%% End: 

\section{Evaluation}
\label{Se:eval}
In this section, we evaluate the effectiveness and performance of \name, and compare it to existing tools.\footnote{All the experiments were performed on an Intel Core i7 4.00GHz CPU, with
	8GB of RAM. We used version 4.8.9 of Z3.
	The timeout for each benchmark was 10 minutes.}
	We implemented \name in Haskell on top of \synquid by extending its type system with recurrence annotations as presented in \sectref{typeSystem}. The detailed results are reported in \sectref{aEval}.

\subsection{Comparison to Prior Tools}

We compared \name against two related tools: \synquid\cite{polikarpova2016program} and \resyn\cite{knoth2019resource},
which are also based on refinement types.

\mypar{Benchmarks.}
We considered a total of \numAllBenchmarks synthesis problems:
56 synthesis problems from \resyn (each benchmark specifies a concrete linear-time resource annotation), 16 synthesis problems from \synquid (which do not include resource annotations) that are not included in \resyn, and
 \numMoreBenchmarks new synthesis problems involving non-linear resource annotations. In these synthesis problems, synthesis specifications and auxiliary functions are all given as refinement types. 
For 3 of the new benchmarks, the auxiliary function required to split the input into smaller ones is not given---i.e.,
the synthesizer needs to identify it automatically. 

The three solvers (\tool, \synquid, and \resyn) have different features,
and hence not all synthesis problem can be encoded as synthesis benchmarks for a single solver.
In the rest of this section, we describe what benchmarks we considered
for each tool, and how we modified the benchmarks when needed.

%There 77/\numAllBenchmarks synthesis problems that can be encoded to \synquid benchmarks (synthesis without resource bounds), 57/\numAllBenchmarks synthesis problems that can be encoded to \resyn benchmarks (synthesis with linear concrete resource bounds), and 70/\numAllBenchmarks synthesis problems that can be encoded to \tool benchmarks (synthesis with asymptotic resource bounds). The pattern we used for the rule \textsc{T-Abs} in our evaluation are those shown in \tableref{pattern}, the one shown in \sectref{extension}, and the simple non-recursive pattern $T(\emph{u})\le O(\psi)$ implying $T(\emph{u})\in O(\psi)$.

\mypar{\synquid:} \synquid does not support resource bounds, so
we encoded \numAllBenchmarks synthesis problems as \synquid benchmarks by dropping the resource annotations.
\synquid returns the first program that meet the synthesis specification, and 
cannot provide any guarantees about the resource usage of the returned program. 
\synquid can solve \numBenchmarksSolvedBySynquid benchmarks, and takes on average 3.3s.
For 10 benchmarks \synquid synthesizes a non-optimal program---i.e., there exists another program with better concrete resource usage.
For example, on the \resyn-triple-2 benchmark (where the input is a list $xs$), \synquid found a solution with resource usage $O(|xs|^2)$, while both \tool and \resyn can synthesize a more efficient implementation with resource usage $O(|xs|)$.
	The two benchmarks that \synquid failed to solve include the new benchmark \tool-merge-sort'. In this benchmark, the auxiliary function required to break the input into smaller inputs is not given, without which the sizes of solutions become much larger. Therefore \synquid times out.

\mypar{\resyn:}
We ran \resyn on the 56 \resyn benchmarks with the corresponding concrete resource bounds.
We could not encode 16 problems because \resyn does not support non-linear resources
bounds---e.g., the bound $\log |y|$ in the AVL-insert \synquid benchmark. 
\resyn solved all 56 benchmarks with an average running time of 18.3s. 

\mypar{\tool:} We manually added resource usages and resource bounds to existing problems to encode them for \tool.
For \synquid benchmarks without concrete resource bounds, we chose
well-known time complexities as the bounds, e.g., we added the
resource bound $O(\emph{u}\log\emph{u})$ to the Sort-merge-sort
problem.
For the \resyn benchmarks, we translated the concrete resource usage
and resource bounds to the corresponding asymptotic ones---e.g., for
the \resyn-common' benchmark with the concrete resource bound
$|ys|+|zs|$, we constructed a \tool variant with the asymptotic bound
$O(\emph{u})$ and a size function $\lambda ys.\lambda zs.|ys|+|zs|$. 
We could not encode 9 synthesis problems as \tool benchmarks because
they involved higher-order functions, which are not supported by \tool, or the resource usage bound $O(2^\emph{u})$ (the Tree-create-balanced problem from \synquid).

\tool solved 68 benchmarks with an average running time of 8.1s. 
Unlike \synquid, \tool guarantees that the synthesized program satisfies the given resource bounds.
For 10 benchmarks, \tool found programs that had better resource usage than those synthesized by \synquid.
Furthermore, \tool can encode and solve $9$ problems that \resyn could not solve because the resource bounds involve logarithms. However, \tool cannot encode and solve $8$ benchmarks that involve higher-order functions.
\tool could solve 3 problems that required synthesizing both the main function (e.g., \tool-merge-sort) and its auxiliary function (e.g., the function splitting a given list into two balanced partitions). 
No other tool could solve the \tool-merge-sort' benchmark.

%Also, for problems that can be encoded by both \resyn and \tool,  different types of resource bounds  are required as inputs the each tools. For example, in the benchmark \resyn-triple, users will need to know the constant 2 to specify the concrete resource bound $2|xs|$, while it is easier to specify an asymptotic bound $O(|xs|)$. 

\begin{mdframed}[innerleftmargin = 3pt, innerrightmargin = 3pt, skipbelow=-0.25em]
\mypar{Finding.}
\tool can express and solve 68/\numAllBenchmarks benchmarks.
\tool has comparable performance to existing tools, and can
synthesize programs with resource bounds that are not supported by prior tools.
\end{mdframed}

\subsection{Pruning the Search Space with Annotated Types}
\label{Se:searchSpace}

\tool uses recurrence annotations to guide the search and avoids enumerating terms that are guaranteed to not 
match the specified complexity. 
We compared the numbers of E-terms enumerated by \tool and \synquid for 56 benchmark on which both tool produced same solutions.  \synquid always enumerated at least as many E-terms as \tool, and
\tool enumerated strictly fewer E-terms for 26/56 benchmarks. 
For these 26 benchmarks, \tool can on average prune the search space by 6.2\%.
For example, in one case (BST-delete) \tool enumerated 2,059 E-terms, while \synquid enumerated 2,202.
%The difference of numbers of enumerated E-terms shows how the annotated type system helps to further prune the search space  comparing to refinement types.

\begin{mdframed}[innerleftmargin = 3pt, innerrightmargin = 3pt, skipbelow=-0.25em]
\mypar{Finding.} On average, \tool reduces the size of the search space by 6.2\% for approximately half of the benchmarks.
\end{mdframed}

%%% Local Variables: 
%%% mode: latex
%%% TeX-master: "main_cav.tex"
%%% End: 
	
\section{Related Work}
\label{Se:RelatedWork}

\mypar{Resource-Bound Analysis.}
Rather than determining whether a given program satisfies a specification,
a synthesizer determines whether there exists a program that inhabits a
given specification.
The branch of verification that we draw upon for resource-based
synthesis is resource-bound analysis \cite{RBA:Wegbreit75}.

Within the literature on automated resource-bound analysis, there are methods that extract and solve recurrence relations for imperative code \cite{PUBS:AAGP11,Thesis:FloresMontoya,PLDI:BCKR20,PACMPL:KCBR18}. However, these methods are unlike the type system presented in this work because they extract concrete complexity bounds as recurrence relations, and then solve the recurrences to find a concrete upper bound on resource usage. The dominant terms of the resulting concrete bounds can then be used to state a big-$O$ complexity bound. 
In contrast, we want to synthesize programs with respect to a big-$O$ complexity directly, which is more similar to the manual reasoning of \cite{eberl2017proving,gueneau2018fistful}.
Thus, if we were to use these techniques for our problem, the first step in our synthesis algorithm
would be to pick a concrete complexity function given a big-$O$ complexity, and then
reverse the verification problem with regards to that concrete complexity.
However, for any big-$O$ complexity, there are an infinite number of functions that satisfy
that complexity, which presents a significant challenge at the outset.
Our design choice also has some drawbacks.
As noted in \cite{gueneau2018fistful}, reasoning compositionally with big-$O$ complexity
is challenging due to the hidden quantifier structure of big-$O$ notation.
Thus, to maintain soundness our type system has to sacrifice precision and generality in some places.
For example, when a function has multiple paths, our type system over-approximates by choosing
the largest complexity among all the paths.

Another set of methods to generate resource bounds are type-based \cite{hoffmann2011multivariate,CAV:HAH12,OOPSLA:WWC17,FoSSaCS:KH20}. 
As we discussed throughout the paper, the complexities generated by these methods are concrete functions and not expressed with big-$O$ notation,
although \cite{OOPSLA:WWC17} is sometimes able to pattern match a case of the Master Theorem.
These type systems differ from ours in a few ways. The AARA line of research \cite{hoffmann2011multivariate,CAV:HAH12,FoSSaCS:KH20} is able to assign amortized complexity to programs, but is not able to generate logarithmic bounds.
\cite{OOPSLA:WWC17} is also able to perform amortized analysis; however, the technique is not fully automated,
and instead requires the user to  provide type annotations on terms,
which are then checked by the type system.

\mypar{Type- and Resource-Aware Synthesis.}
The \name implementation is built on top of \synquid~\cite{polikarpova2016program}
a type-directed synthesis tool based on refinement types and polymorphism. 
The work that most closely resembles ours is \resyn
\cite{knoth2019resource}.
As in our work, they combine the type-directed synthesizer \synquid
with a type system that is able to assign complexity bounds to
functional programs.
The type system used in \resyn is based on one originally
used in the context of verification \cite{CAV:HAH12}.
That work uses a sophisticated type system to assign amortized
resource-usage bounds to a given program.
The type system of \resyn differs from the one presented in
\sectref{typeSystem} in a few significant ways.

As highlighted earlier, \resyn  automatically infers bounds on recursive functions using  amortized analysis  and is restricted
to linear bounds, whereas our system is able to synthesize complexities of the 
form $O(n^a\log^b n+c)$. 

Another difference is that \resyn synthesizes programs with a
concrete complexity bound.
This approach has advantages and disadvantages.
For instance, it places an extra burden on the human to provide the
correct bound with precise coefficient.
On the other hand, the user might want an implementation that has a
complexity with a small coefficient, whereas our system provides no
guarantee that the complexity of an implementation will have a small
coefficient in the dominant term: \tool only guarantees asymptotic behavior.

\resyn can synthesize programs with higher-order functions, which are not supported by \tool. 
To handle higher-order functions, \resyn attaches resource units to types, which gives it \emph{resource polymorphism}.
Moreover, costs of inputs with function types can be written generally as polymorphic types (i.e., costs can be polymorphic with respect to the size of the specific input types).
\tool does not have \emph{asymptotic resource polymorphism} because it cannot directly compose unknown big-$O$ functions (i.e., the complexity of higher-order inputs). 
We envision that with carefully crafted restrictions on the resource annotations of higher-order functions, \tool could handle synthesis problems involving such functions, e.g., assuming that the complexity of input functions is known and the refinements of input functions are precise enough. 
%We also observe that all the higher-order benchmarks from the \resyn paper
Because big-$O$ functions cannot be directly composed,
developing a more general extension to \tool that supports higher-order functions is a challenging direction for future work.

\mypar{Acknowledgments.}
Supported, in part,
by a gift from Rajiv and Ritu Batra;
by multiple Facebook Research Awards;
by a Microsoft Faculty Fellowship;
by NSF under grants 1420866, 1763871, and 1750965; 
and by ONR under grants N00014-17-1-2889 and N00014-19-1-2318.
Any opinions, findings, and conclusions or recommendations
expressed in this publication are those of the authors,
and do not necessarily reflect the views of the sponsoring
entities.

%%% Local Variables: 
%%% mode: latex
%%% TeX-master: "main_cav.tex"
%%% End: 

\bibliographystyle{splncs04}
\bibliography{ref}
\newpage
\appendix
%\nochangebars
%\input{a1_syntax}
\section{Semantics}
\label{Se:aSemantics}
In this section, we presented two kinds of semantics: 1) the concrete small-step semantics which define the concrete complexity of functions, and 2) a loose semantics which over-approximates the concrete semantics and will be used in the proof of the soundness theorem.

\mypar{Concrete semantics.} The evaluation rules of concrete-cost semantics are shown in \figref{concrete}.
In the concrete-cost semantics, a configuration $\concreteConfig{t,C}$
consists of a term $t$ and a nonnegative integer $C$ denoting the
resource usage so far.
The evaluation judgment
$\concreteConfig{t,C} \hookrightarrow \concreteConfig{t',C+C_\Delta}$ states
that a term $t$ can be evaluated in one step to a term (or a value) $t'$, with
resource usage $C_\Delta$.
We write $\concreteConfig{t,C} \hookrightarrow^* \concreteConfig{t',C+C_\Delta}$ to indicate
the reduction from $t$ to $t'$ in zero or more steps.

\mypar{Polynomial-Bounded Refinement Type.}
Before introducing the loose semantics, we introduce a class of refinement type that can be over-approximated as polynomials.

The resource usage of a function call depends on the size of the problem it makes calls to. However, predicates used to refine functions could be imprecise, and thus we may have to reason about sizes of problems with imprecise refinements. For example, consider a function $\texttt{square::}\langle\texttt{Int}\to\{\texttt{Int}~|~\emph{v}\ge0\},\generatecolor{(O( \emph{u}))}\rangle$. In the nested application  term $\texttt{square}~(\texttt{square}~\texttt{x})$, we know the size of the problem of the inner application is $\texttt{x}$, but we do not know the size of the problem $\texttt{square~x}$ of the outer application (all we know about $\texttt{square~x}$ is that it is non-negative).
To reason about input sizes, we introduce a  class of refinement types with which we can infer upper bounds of the sizes of problems.

Let us first assume that all terms and variables  are integers or integer functions (we will discuss the case where terms and variables hold values other than integers later). For an integer refinement type $\tau=\{\texttt{Int}~|~\varphi\}$ refined by a predicate $\varphi(\emph{v},\overline{x})$ over $\emph{v}$ and a tuple of variables $\overline{x}$, we say that $\varphi$ is \emph{bounded} by a \emph{polynomial} term $p(\overline{x})$, written as $\varphi\bound p$ if 
$$\exists\overline{c}>\overline{0}\forall\overline{x}\forall\emph{v}.~\overline{x}\!>\!\overline{c}\wedge\varphi(\emph{v},\overline{x})\implies |\emph{v}|<p(\overline{x}).$$
That is, there exist some positive constants $\overline{c}$ such that, for any $\overline{x}$ (point-wise) greater than $\overline{c}$ and for any $\emph{v}$, if $\emph{v},\overline{x}$ satisfying $\varphi$, the absolute value $|\emph{v}|$ is always less than $p(\overline{x})$. 
\iffalse
Similarly, an  arrow type $\tau_\texttt{f}::x\!:\!\tau_1\!\to\!..\!\to\!x_n\!:\!\tau_n\to\tau$ is bounded by a polynomial $p$ if $\tau$ is bounded by $p$ with the refinements of arguments as context:
$$\exists\overline{c}.\forall\overline{x}.\forall\emph{v}.\overline{x}\!>\!\overline{c}\wedge\varphi(\emph{v},\overline{x})\wedge\bigwedge_{i=1}^n[x_i/v]\varphi_i\implies \emph{v}<p(\overline{x}), $$
where $\varphi_i$ is the refinement in $\tau_i$. 
\fi
\begin{example}
	The refinement type $\tau:=\{\texttt{Int}~|~\emph{v}\le 2x+y\vee\emph{v}\le x+2y\}$ is bounded by the expression $p_1:=2x+2y$ and the expression $p_2:=3x+3y$ but not $p_3:=2x+y+1$, i.e., $\tau\bound p_1$,~$\tau\bound p_2$, and $\tau\not\sqsubset p_3$.
\end{example}

For datatype $D$ other than $\texttt{Int}$, we assume that there is an \emph{intrinsic measure} with output type $\texttt{Int}$ for every $D$, denoted by $|\!\cdot\!|_D$ (we omit the subscript if it is clear from the context). Intrinsic measures are specified by users for user-defined datatypes. For example, the intrinsic measure of lists can be defined as a function that computes the length of lists, i.e., $|l|=\texttt{len}~l$ for any list $l$. The instinct measure of $\texttt{Int}$ term is the absolute-value function. The condition of $p$ bounding $\tau=\{D~|~\varphi(\emph{v},\overline{x})\}$ becomes
$$\exists\overline{c}>\overline{0}\forall\overline{x}\forall\emph{v}~\overline{|x|}\!>\!\overline{c}\wedge\varphi(\emph{v},\overline{x})\implies |\emph{v}|<p(\overline{|x|}). $$

\begin{figure}[t]
	\begin{mathpar}
		{\inferrule*[right=\textsc{Sem-Tick }]
			{ }
			{\concreteConfig{\texttt{tick}(c,t),C} \hookrightarrow \concreteConfig{t,C+c}}
		}\\
		{\inferrule*[right=\textsc{Sem-App }]
			{ \overline{v}~\text{are values}}
			{\concreteConfig{(\texttt{fix}~f.\lambda \overline{x}.t)\overline{v},C} \hookrightarrow \concreteConfig{[(\texttt{fix}~f.\lambda \overline{x}.t)/f][\overline{v}/\overline{x}]t,C}}
		}\\
		{\inferrule*[right=\textsc{Sem-App-Arg }]
			{ \concreteConfig{t_1,0} \hookrightarrow \concreteConfig{t_2,C_\Delta}}
			{\concreteConfig{(\texttt{fix}~f.\lambda \overline{x}.t)t_1,C} \hookrightarrow \concreteConfig{(\texttt{fix}~f.\lambda \overline{x}.t)t_2,C+C_\Delta}}
		}\\	
		{\inferrule*[right=\textsc{Sem-Cond-True }]
			{ }
			{\concreteConfig{\texttt{if}~\texttt{true}~\texttt{then}~t_1~\texttt{else}~t_2,C} \hookrightarrow \concreteConfig{t_1,C}}
		}\\{\inferrule*[right=\textsc{Sem-Cond-False }]
			{ }
			{\concreteConfig{\texttt{if}~\texttt{False}~\texttt{then}~t_1~\texttt{else}~t_2,C} \hookrightarrow \concreteConfig{t_2,C}}
		}\\{\inferrule*[right=\textsc{Sem-Cond-Guard}]
			{ \concreteConfig{t_c,0} \hookrightarrow^* \concreteConfig{b,C_{\Delta}}\quad b~\text{is a Boolean value}}
			{\concreteConfig{\texttt{if}~t_c~\texttt{then}~t_1~\texttt{else}~t_2,C} \hookrightarrow \concreteConfig{\texttt{if}~b~\texttt{then}~t_1~\texttt{else}~t_2,C+C_\Delta}}
		}\\{\inferrule*[right=\textsc{Sem-Match}]
			{ \overline{v}~\text{are values}}
			{\concreteConfig{\texttt{match}~\texttt{C}_j(\overline{v})~\texttt{with}~|_i\texttt{C}_i(\overline{x}_i) \mapsto  t_i,C} \hookrightarrow \concreteConfig{[\overline{v}/\overline{x}_i]t_i,C}}
		}\\{\inferrule[\textsc{Sem-Match-Scrutinee}]
			{ \concreteConfig{t_s,0} \hookrightarrow^* \concreteConfig{\texttt{C}_j(\overline{v}),C_{\Delta}}\quad \overline{v}~\text{are  values}}
			{\concreteConfig{\texttt{match}~t_s~\texttt{with}~|_i\texttt{C}_i(\overline{x}_i) \mapsto  t_i,C} \hookrightarrow \concreteConfig{\texttt{match}~\texttt{C}_j(\overline{v})~\texttt{with}~|_i\texttt{C}_i(\overline{x}_i) \mapsto  t_i,C+C_\Delta}}
		}
	\end{mathpar}
	\caption{Evaluation rules of the concrete small-step semantics.}
	\label{Fi:concrete}
\end{figure}

\mypar{A loose cost model.}  
The semantics given previously give a standard notion of complexity. However, we find two challenges connecting these semantics to our synthesis algorithm. First, we allow users to supply auxiliary functions as signatures involving big-O notation as opposed to implementations. Second, our synthesis algorithm ensures complexity through the tracking of recursive calls, which are not present in the concrete semantics given above. To address these challenges we introduce an intermediate semantics that uses recurrence relations and big-O notation. We then show in \theoref{looseIsSound} that this intermediate semantics approximates complexity in the sense of \defrefs{complexity}{bigO}.

The signatures of auxiliary functions $g$ are of the form $\langle\tau_1\to\{B~|~\varphi(\emph{v},\overline{y})\},\generatecolor{O(\psi(\emph{u}))}\rangle$.  Although we don't really have the implementation of $g$, we assume that there exists some implementation $\texttt{fix}~g.\lambda \overline{y}.t$ of $g$, such that
\begin{itemize}
	\item for any input $\overline{x}$, the output of $g$ on $\overline{x}$ satisfies the signature, i.e., $\concreteConfig{(\texttt{fix}~g.\lambda \overline{y}.t)\overline{x},0}\hookrightarrow^*\concreteConfig{v_{\overline{x}},C_{\overline{x}}}$ implies $\varphi(v_{\overline{x}},\overline{x})$; and
	\item  for any input $\overline{x}$, the complexity  of $g$ is bounded by $\psi(\emph{u})$, i.e., $T_g(n)\in O(\psi(n))$.
\end{itemize}
For the top-level function we are evaluating, we assume that its signature $\langle\tau_1\to\{B~|~\varphi(\emph{v},\overline{y})\},\generatecolor{O(\psi(\emph{u}))}\rangle$ is also given, whereas the semantics of $f$ is over-approximated by its refinement, i.e., for any input $\overline{x}$,  $\concreteConfig{(\texttt{fix}~f.\lambda \overline{y}.t)\overline{x},0}\hookrightarrow^*\concreteConfig{v_{\overline{x}},C_{\overline{x}}}$ implies $\varphi(v_{\overline{x}},\overline{x})$.

Now we introduce our intermediate loose semantics. Formally, reductions are defined between configurations. Each configuration $\config{\hat{t},\checkcolor{\mathcal{R}}}$ is a pair of an extended term $\hat{t}$ and a recurrence parameter $\mathcal{R}$. Extended terms $\hat{t}::=t~|~\varphi$ are either terms $t$ or  formula expressions $\varphi$. Recurrence parameters $\mathcal{R}::=\phi~|~\bot~|~\mathcal{R}\!\parallel\!\mathcal{R}$ are either size expressions $\phi$, or parameters combined by a \emph{parallel} operator $\parallel$, i.e., $\mathcal{R}$ is a collection of size expressions.    
The parallel operator $\parallel$ distributes over the plus, i.e., $(\mathcal{R}_1\parallel\mathcal{R}_2)+\mathcal{R}_3=\mathcal{R}_1+\mathcal{R}_2\parallel \mathcal{R}_1+\mathcal{R}_3$. Intuitively, a parameter $\mathcal{R}$ without parallelism denotes the recurrence relation of the function along one path. When the function contains more than one path, the overall parameter will be sub-parameters in parallel.

We use $\config{t,\checkcolor{\mathcal{R}}}\mapsto\config{\varphi,\checkcolor{\mathcal{R}'}}$ to denote a step of a reduction.  The goal is to reduce a term $t$ in a function $\texttt{f}$ to a predicate $\varphi$ such that $\varphi$ describes the behavior of $t$---$\varphi$ is a refinement of $t$. At the same time,  recurrence relations can be built  by incrementally appending expressions representing  resource usage to the recurrence parameter $\checkcolor{\mathcal{R}}$.

Because we are building recurrence relations for a function $\texttt{f}$, the reduction always starts from a $\texttt{\textbf{fix}}$-term and an empty parameter $\bot$. The result configuration is the refinement $\varphi$ of the function body $t$  with the recurrence parameter $\checkcolor{ \mathcal{R}}$. We use the function $T:\texttt{Int}\to\texttt{Int}$ to denote the resource usage of the function $\texttt{f}$; hence, the recurrence parameter we build for $\texttt{f}$ will be the recurrence relation of resource usage $T$.
$${\inferrule*[right=\textsc{LooseSem-Fix }]
	{\config{t,0}\mapsto \config{\varphi,\checkcolor{\mathcal{R}}} }
	{\config{\texttt{\textbf{fix}~f}.\lambda x_1..\lambda x_n.t,~\bot}\mapsto \config{\varphi,~\checkcolor{ \mathcal{R}}}}
}$$ 

In our loose semantics, each auxiliary function $g$ has a resource annotation $\generatecolor{(O(\psi_g))}$ denoting the resource usage of $g$, a logical signature $\varphi_g$ denoting the behavior of $g$, and a size function $\texttt{size}_g$. Resource usage happens  when an auxiliary function is called.
$${\inferrule*[right=\textsc{LooseSem-App}]
	{
		g\!:\!\langle x_1\!:\!\tau_1\!\to\!..\!\to\!x_n\!:\!\tau_n\!\to\!\{B~|~\varphi_g\},\generatecolor{O(\psi_g)} \rangle\\  \forall i.\config{e_i,0}\mapsto\config{\varphi_i,\checkcolor{\mathcal{R}_i}}\qquad\varphi:= \varphi_g\wedge\bigwedge_{i=1}^n[x_i/\emph{v}]\varphi_i\quad
		\forall i.\varphi_i\bound v_i}
	{\config{g~e_1..~e_n,0}\mapsto \config{\varphi,\checkcolor{~[(\texttt{size}_g~v_1..~v_n)/\emph{u}]\psi_g+\sum_{i=1}^n\mathcal{R}_i}}}
}$$
That is, if each callee $e_i$ can be reduced to a predicate $\varphi_i$ bounded by $v_i$ with the recurrence parameter change $\mathcal{R}_i$, the non-recursive application term  $g~e_1..~e_n$ will be reduced to the predicate $\varphi$ with resource usage  $\checkcolor{[\texttt{size}_g~v_1..~v_n/\emph{u}]\psi_g}$ (the resource usage of $g$) and $\checkcolor{\sum_{i=1}^n\mathcal{R}_i}$ (resource usage used to evaluate $\{e_i\}_i$). We over-approximate the size of the problem $e_1..e_n$ using the upper bounds $v_i$ of callee's behavior predicates $\varphi_i$. The result predicate $\varphi$ is actually the behavior $\varphi_g$ of $g$ with each argument $x_i$ instantiated with the semantic predicates $[x_i/\emph{v}]\varphi_i$ of callee $e_i$.

The semantics of performing a recursive call is a bit different. The resource usage is instead $\checkcolor{T(\texttt{size}_\texttt{f}~v_1..~v_n)}$---the resource usage $T$ on a sub-problem with size  $\texttt{size}_\texttt{f}~v_1..~v_n$ where the $v_i$'s are over-approximations of the callee $e_i$'s.
$${\inferrule[\textsc{LooseSem-RecApp }]
	{ \forall i.\config{e_i,0}\mapsto^\config{\varphi_i,\checkcolor{\mathcal{R}_i}} \qquad 
		\forall i.\varphi_i\bound v_i 	\\
		\texttt{f}\!:\!\langle x_1\!:\!\tau_1\!\to\!..\!\to\!x_n\!:\!\tau_n\!\to\!\{B~|~\varphi_\texttt{f}\},\generatecolor{O(\psi_\texttt{f})} \rangle\qquad\varphi:= \varphi_\texttt{f}\wedge\bigwedge_{i=1}^n[x_i/\emph{v}]\varphi_i}
	{\config{\texttt{f}~e_1..~e_n,0}\mapsto \config{\varphi,~\checkcolor{T(\texttt{size}_\texttt{f}~v_1..~v_n)+\sum_{i=1}^n\mathcal{R}_i}}}
}$$

The reduction of \texttt{if}-terms will result in $\texttt{ite}$ predicates. The resulting recurrence parameter $\checkcolor{\mathcal{R}_e+\mathcal{R}_1\parallel\mathcal{R}_e+\mathcal{R}_2}$ uses the parallel operator because there are two paths in an \texttt{ite} term.
$${\inferrule[\textsc{LooseSem-Cond }]
	{\config{e,0}\mapsto\config{\varphi_e,\checkcolor{\mathcal{R}_e}}\quad\config{t_1,0}\mapsto\config{\varphi_1,\checkcolor{\mathcal{R}_1}}\quad\config{t_2,0}\mapsto\config{\varphi_2,\checkcolor{\mathcal{R}_2}}\\ \varphi:=\varphi_1\vee\varphi_2}
	{\config{\texttt{\textbf{if}}~e~\texttt{\textbf{then}}~t_1~\texttt{\textbf{else}}~t_2,0}\mapsto \config{\varphi,~\checkcolor{\mathcal{R}_e+\mathcal{R}_1\parallel\mathcal{R}_e+\mathcal{R}_2}}}
}$$

The rules for match term and variable term are similar.
$${\inferrule[\textsc{LooseSem-Match}]
	{\forall i.~\config{t_i,0}\mapsto\config{\varphi_i,\checkcolor{\mathcal{R}_i}}\quad\varphi:=\varphi_1\vee\ldots\vee\varphi_m}
	{\config{\texttt{\textbf{match}}~e~\texttt{\textbf{with}}~|_i~\texttt{C}_i~(x_i^1 \ldots x_i^n)\mapsto t_i ,0}\mapsto \config{\varphi,~\checkcolor{\mathcal{R}_e+\mathcal{R}_1\parallel\ldots\parallel\mathcal{R}_e+\mathcal{R}_m}}}
}$$
$${\inferrule*[Right=\textsc{LooseSem-Var }]
	{\ }
	{\config{x,0}\mapsto \config{|\emph{v}|=|x|,~0}
}}$$
With the above rules, we can then say the complexity of a function to be any expression that satisfy the recurrence parameter of the function.

\begin{theorem}[Complexity bounds]\label{The:looseIsSound}
	Given a function term $\texttt{\textbf{fix}~f}.\lambda \overline{x}.t_{\texttt{f}}$, the signature type of $f$, and the signature types of all auxiliary function used in $f$, if 
	\begin{itemize}
			\item the refinements of  auxiliary functions and $f$ are all bounded by some monotonic non-decreasing polynomials; 
		\item the function body $t_{\texttt{f}}$ can be reduced to $\config{\cdot,\checkcolor{\mathcal{R}_{\texttt{f}}}}$, where $\checkcolor{\mathcal{R}_{\texttt{f}}}$ is of form $\checkcolor{\mathcal{R}_{\texttt{f},1}\parallel..\parallel\mathcal{R}_{\texttt{f},m}}$, and none of the $\checkcolor{\mathcal{R}_{\texttt{f},i}}$'s contain an occurrence of the parallel operator; and
		\item there exists a function $\psi$ that satisfies
\[
  \forall i.~T(\texttt{size}_\texttt{f}(\overline{x}))\le\checkcolor{\mathcal{R}_{\texttt{f},i}}\implies T(\texttt{size}_{\texttt{f}}(\overline{x}))\in O(\psi(\texttt{size}_{f}(\overline{x}))),
\]
	\end{itemize} then the complexity $T_f$ of $\texttt{f}$ is bounded by the function $\psi$, i.e., $T_f\in O(\psi)$.
\end{theorem}
\begin{proof}
	\begin{changebar}
		
		We first show by induction on a loose semantic derivation that for any term $t$, if
                \begin{itemize}
			\item		$\config{t,0}\mapsto\config{\varphi_t,\checkcolor{\mathcal{R}}}$, where $\checkcolor{\mathcal{R}}$ is of form $\checkcolor{\mathcal{R}_1}\parallel..\parallel\checkcolor{\mathcal{R}_l}$, and none of the $\checkcolor{\mathcal{R}_i}$'s contains an occurrence of the parallel operator; and
			\item $\concreteConfig{[(\texttt{fix}~f.\lambda\overline{x}.t_f/f][(\texttt{fix}~g.\lambda\overline{y}.t_g/g][\overline{in}/\overline{x}]t,0}\hookrightarrow^*\concreteConfig{v_t,C_t}$ for all $\overline{in}$, 
		\end{itemize}
                then we have for any $\overline{in}$, $[v_t/v][\overline{in}/\overline{x}]\varphi_t$ is satisfiable---that is, the loose semantics $\varphi_t$ over-approximates the concrete semantics $v_t$---and $\exists   c,k>0~\forall \overline{in}~\exists j.~\texttt{size}_f(\overline{in})>{c}\implies C_t<k*\checkcolor{\mathcal{R}_j}$.
		
		%for any $\overline{in}$ and polynomial $p_t$, $\varphi_t\bound p_t\implies |v_t|\le p_t(\overline{|in|})$, and $\exists   c,k>0~\forall \overline{in}~\exists j.~\texttt{size}_f(\overline{in})>{c}\implies C_t<k*\checkcolor{\mathcal{R}_fsj}$.
		
		\textbf{Base case.} The base case is when $\varphi$ is a variable term. The $\varphi$ for a variable term is precise, and the resource usage are 0 in both semantics. 
		
		\textbf{Non-recursive application terms.} When $t=g~e_1..~e_n$ is a non-recursive application term, with the induction hypothesis, we have that the loose semantics $\varphi_i$ for each $e_i$ over-approximates the concrete semantics $v_{e_i}$, i.e., $[v_{e_i}/v][\overline{in}/\overline{x}]\varphi_{e_i}$ is satisfiable for all $e_i$ and $\overline{in}$.
		According to the first assumption of the signature of auxiliary functions, the predicate $\forall \overline{in}.~\varphi_g(v_{t},\overline{v}_{e_i})$ is valid. The loose semantics of $t$ is 
		$\varphi_t:=\varphi_g(v,\overline{y})\wedge_{i=1}^n[y_i/v]\varphi_{e_i}.$
		Finally, the predicate
		\begin{eqnarray*}
			\forall\overline{in}. [v_t/v][\overline{in}/\overline{x}]\varphi_t =\varphi_g(v_t,\overline{y})\wedge_{i=1}^n[y_i/v][\overline{in}/\overline{x}]\varphi_{e_i}
		\end{eqnarray*}
		is satisfiable with the assignment $y_i\gets v_{e_i}$ for all $i$.
		
		According to the second assumption on signatures of auxiliary functions, we have $T_g(n)\in O(\psi_g(n))$. That is, $\exists \overline{c},k~\forall\overline{in}.~\texttt{size}_g(\overline{in})>\overline{c}\implies C_t\le k*\psi_g(\texttt{size}_g(\overline{|v_{e_i}|}))$. Then we have $\exists \overline{c},k\forall\overline{in}.~\texttt{size}_g(\overline{in})>\overline{c}\implies C_t\le k*\checkcolor{~\psi_g((\texttt{size}_g~v_1..~v_n))}$  because each $v_i$ is the bound of $\varphi_i$ and hence the bound of the concrete value $v_{e_i}$.

		\textbf{Recursive application terms.} When $t=f~e_1..~e_n$ is a recursive application term, the proof of the behavior part is similar as above because we have the same behavior assumption on the signature of the top-level function. The concrete cost $C_t$ of $t$ is bounded by $T_f(\texttt{size}_f(|v_{e_1}|,\ldots,|v_{e_n}|))$, where $T_f$ is the complexity function of $f$, plus the concrete cost of evaluating each $e_i$ (which is bounded by $\checkcolor{\sum \mathcal{R}_i}$ according to the induction hypothesis). Note that here $T_f$ is an uninterpreted, monotonic, non-decreasing, non-negative function. The loose semantics $\checkcolor{T(\texttt{size}_f(~v_1..~v_n))}$ also contains an uninterpreted, monotonic, non-decreasing, non-negative function $T$. In this proof, we generalize the comparison symbol $\le$ to 
		$T_f(n)\le T(n')$ if $n\le n'$, that is, the comparison between  uninterpreted functions is the result of comparison between their inputs. With such generalization,
		$T_f(\texttt{size}_f(|v_{e_1}|,\ldots,|v_{e_n}|))\le\checkcolor{T(\texttt{size}_f(~v_1..~v_n))}$ because each $v_{e_i}$ is bounded by $v_i$ according to the induction hypothesis.
		
		\textbf{Branching terms.} When $t=\texttt{\textbf{if}}~e~\texttt{\textbf{then}}~t_1~\texttt{\textbf{else}}~t_2$ is a conditional term, the concrete cost $C_{t_1}+C_e$ or $C_{t_2}+C_e$ of it is bounded by the loose cost $\checkcolor{\mathcal{R}_e+\mathcal{R}_1}$ or $\checkcolor{\mathcal{R}_e+\mathcal{R}_2}$, respectively, according to the induction hypothesis.  The concrete semantics is either $v_{t_1}$ or $v_{t_2}$. According to the induction hypothesis, $\varphi_{t_1}\vee\varphi_{t_2}$ over-approximates both branches.
		
		The case of match term is similar to the case for conditional terms.
		
		Now, for any input $\overline{in}$, the complexity function $T_f$ of the top-level function $f$ should  satisfy the recurrence parameter along one of the paths, i.e., 
		$\exists k\forall \overline{in}~\exists j.T_f(\texttt{szie}_f(\overline{in}))\le k*[T_f/T]\checkcolor{\mathcal{R}_{f,j}}.$  So, if a bound $\psi$ dominates the loose cost for every path, it will always dominate the complexity $T_f$ of $f$.
	\end{changebar}
\end{proof}

\begin{example}
	For the program shown in 
	\eqref{prod_log}, there are three paths. At the beginning, $\texttt{\textbf{fix}~prod}.\lambda x.\lambda y.t $ is reduced to $\config{\varphi,\checkcolor{T(\texttt{size}_{\texttt{\textbf{prod}}}~x~y)\le 0+\mathcal{R}}}$ where $\mathcal{R}$ is the reduction result of the function body, and $\texttt{size}_{\texttt{\textbf{prod}}}~x~y=x$ since the size function for \texttt{prod} is $\lambda z.\lambda w.z$
	
	The first \texttt{ite} term $\texttt{\textbf{if}}~x==0~\texttt{\textbf{then}}~t_1~\texttt{\textbf{else}}~t_2$ is reduced to $\config{\texttt{ite}(x==0,~\varphi_1,~\varphi_2),\checkcolor{\mathcal{R}_e+\mathcal{R}_1\parallel\mathcal{R}_e+\mathcal{R}_2}}$ where the condition contain one equivalence operator ($\mathcal{R}_e=1$), the \texttt{then} branch has 0 resource usage (it is a variable term) ($\mathcal{R}_1=0$), and the \texttt{else} branch has resource usage $\mathcal{R_2}$, which we learn from the reduction of $t_2$.
	
	The application term $\texttt{div2~}x$ is reduced to $\config{v=\frac{z}{2}\wedge z=x,~1}$ since the resource  usage of $\texttt{div2}$ is $O(1)$.
	
	The recursive-application term \texttt{prod~(div2~x)~y} is reduced to $\config{v=z*w\wedge z=\frac{x}{2}\wedge w= y,~\mathcal{R}_0+T(\frac{x}{2})+1}$
	
	Overall the recurrence relation can be built as
	$T(x)\le 1\parallel T(x)\le 4 +T(\frac{x}{2})\parallel T(x)\le 5+T(\frac{x}{2})$. Thus the complexity of $\texttt{prod}$ is bounded by $\log x$ since $T(x)\le1\implies T(x)\in O(\log x)$, and, $T(x)\le 5+T(\frac{x}{2})\implies T(x)\in O(\log x)$ according to the Master Theorem.

\end{example}

%%% Local Variables: 
%%% mode: latex
%%% TeX-master: "main_cav.tex"
%%% End: 

\section{Typing rules}
\label{Se:a3Typing}

\begin{figure}[tb]
	\begin{mathpar}
		\boxed{\Gamma\vdash \tau<:\tau'}\hspace{5mm}
		{\inferrule*[right=\textsc{<:-Fun}]
			{\Gamma\vdash \tau_y<:\tau_x\qquad \Gamma;y:\tau_y\vdash[y/x]\tau<:\tau'}
			{ \Gamma\vdash x:\tau_x\to\tau~<:~y:\tau_y\to\tau'}
		}\\
		{\inferrule*[right=\textsc{<:-Sc}]
			{\Gamma\vdash B<:B'\qquad \Gamma\models\varphi\implies\varphi'}
			{ \Gamma\vdash \{B~|~\varphi\}<:\{B'~|~\varphi'\}}
		}\hspace{5mm}
		{\inferrule*[right=\textsc{<:-Refl}]
			{\ }
			{ \Gamma\vdash B~<:~B}
		}\\
		%		{\inferrule*[right=\textsc{<:-Anno}]
		%			{\Gamma\vdash \tau<:\tau'\qquad \Gamma\models [x/\emph{v}]\phi\wedge[x'/\emph{v}]\phi'\Rightarrow \%emph{x}'>\emph{x}\qquad \Gamma\models \psi\in O(\psi')}
		%			{ \Gamma\vdash \langle\tau,\generatecolor{([1,\phi];O(\psi))}\rangle<:\langle\tau',\generatecolor{([1,\phi'];O(\psi'))}}
		%		}
		\\
		{\inferrule*[right=\textsc{<:-Rec}]
			{ c'>c\quad \Gamma\models [x/\emph{v}]\phi\wedge[x'/\emph{v}]\phi'\Rightarrow \emph{x}'>\emph{x}}
			{ \Gamma\vdash \generatecolor{[c,\phi]}<:\generatecolor{[c',\phi']}}
		}\qquad
		{\inferrule*[right=\textsc{<:-Bound}]
			{ \Gamma\models \psi\in O(\psi')}
			{ \Gamma\vdash \generatecolor{O(\psi)}<:\generatecolor{O(\psi')}}
		}\\
		{\inferrule[\textsc{<:-Rec-Split}]
			{  \Gamma\models \phi_1=\phi_2\qquad \Gamma \vdash \tau <:\tau'}
			{  \Gamma\vdash \langle\tau',\generatecolor{([c_1+c_2,\phi_1],[c_3,\phi_3]\ldots[c_n,\phi_n];O(\psi))}\rangle<:\langle\tau,\generatecolor{([c_1,\phi_1],[c_2,\phi_2],\ldots[c_n,\phi_n];O(\psi))}\rangle}
		}\\
		{\inferrule[\textsc{<:-Rec-Combine}]
			{  \Gamma\models \phi_1=\phi_2\qquad \Gamma \vdash \tau <:\tau'}
			{  \Gamma\vdash \langle\tau,\generatecolor{([c_1,\phi_1],[c_2,\phi_2],\ldots[c_n,\phi_n];O(\psi))}\rangle<:\langle\tau',\generatecolor{([c_1+c_2,\phi_1],[c_3,\phi_3]\ldots[c_n,\phi_n];O(\psi))}\rangle}
		}\\{\inferrule*[right=\textsc{<:-Anno-Rec}]
			{\Gamma\vdash \tau<:\tau'\qquad \forall i .\Gamma\vdash \generatecolor{[c_i,\phi_i]}<:\generatecolor{[c'_i,\phi'_i]}\qquad \Gamma\vdash \generatecolor{O(\psi)}<:\generatecolor{O(\psi')}}
			{ \Gamma\vdash \langle\tau,\generatecolor{([c_1,\phi_1],\ldots[c_n,\phi_n];O(\psi))}\rangle<:\langle\tau',\generatecolor{([c_1',\phi_1']\ldots[c'_n,\phi'_n];O(\psi'))}\rangle}
		}
	\end{mathpar}
	\caption{Subtyping rules.}
	\label{Fi:subtype}
\end{figure}
\mypar{Subtyping.} 
	Subtying judgments are shown in \figref{subtype}, and are standard. The most notable rules are \textsc{<:-Rec} and \textsc{<:-Bound}, which states that bounds $\generatecolor{\psi}$ and the sub-problems' size $\generatecolor{\phi}$ in the annotations in subtypes should be less than in the supertype. 
	\[	{\inferrule*[right=\textsc{<:-Rec}]
		{ c'>c\quad \Gamma\models [x/\emph{v}]\phi\wedge[x'/\emph{v}]\phi'\Rightarrow \emph{x}'>\emph{x}}
		{ \Gamma\vdash \generatecolor{[c,\phi]}<:\generatecolor{[c',\phi']}}
	}\\
	{\inferrule*[right=\textsc{<:-Bound}]
		{ \Gamma\models \psi\in O(\psi')}
		{ \Gamma\vdash \generatecolor{O(\psi)}<:\generatecolor{O(\psi')}}
	}\]
	For example, if one branch of some branching term has type $\langle\tau,(\generatecolor{[1,\floor{\frac{\emph{u}}{3}}],O(\psi)})\rangle$, it can be over-approximated by a super type $\langle\tau,(\generatecolor{[1,\floor{\frac{\emph{u}}{2}}],O(\psi)})\rangle$. The idea is that the resource usage of an application calling to a problem of size $\generatecolor{\floor{\frac{\emph{u}}{2}}}$, will be larger than the application calling to a smaller problem of size $\generatecolor{\floor{\frac{\emph{u}}{3}}}$, with the assumption that all resource usages are monotonic. 

 \begin{figure*}[t!]
	\begin{mathpar}
		\boxed{\Gamma\vdash t\irefine\gamma}\hspace{5mm}
		{\inferrule*[right=\textsc{T-Abs}]
			{\Gamma'=[\texttt{recFun}\gets \texttt{f}][\texttt{args}\gets x_1 \ldots x_n]\Gamma\\ 			
				\gamma_f=\langle {x}_1:\tau_{1}\to \ldots \to x_n:\tau_n\to\tau,\generatecolor{(B)}\rangle\\	\Gamma';x_1\!:\!\langle\tau_1,\generatecolor{O(1)}\rangle; \ldots; x_n\!:\!\langle\tau_n,\generatecolor{O(1)}\rangle;\texttt{f}:\gamma_\texttt{f}\vdash t\irefine\langle\tau,\generatecolor{(A)} \rangle
			}
			{ \Gamma\vdash \texttt{fix}~\texttt{f}.\lambda x_1 \ldots \lambda x_n.t\irefine \langle {x}_1:\tau_{1}\to \ldots \to x_n:\tau_n\to\tau,\generatecolor{(B)}\rangle}
		}
		\\{\inferrule*[right=\textsc{T-If}]
			{\Gamma\vdash \alpha\share\alpha_1|\alpha_2\quad  \Gamma \vdash e\eguess \langle\{\texttt{Bool}~|~\varphi_e\},\generatecolor{\alpha_1}\rangle\\
				\Gamma,\varphi_e \vdash t_1\irefine \langle\{B~|~\varphi\},\generatecolor{\alpha_2}\rangle \quad\Gamma,\neg\varphi_e\vdash t_2\irefine \langle\{B~|~\varphi\},\generatecolor{\alpha_2}\rangle
			}
			{ \Gamma\vdash \texttt{\textbf{if}}~e~\texttt{\textbf{then}}~t_1~\texttt{\textbf{else}}~t_2\irefine\langle\{B~|~\varphi\},\generatecolor{\alpha}\rangle
			}
		}	\\
		{\inferrule*[right=\textsc{T-Match}]
			{\Gamma\vdash \alpha\share\alpha_1|\alpha_2\quad  \Gamma \vdash e\eguess \langle\tau_s,\generatecolor{\alpha_1}\rangle\\
				\texttt{C}_i=\tau_1\!\to \ldots \to\!\tau_n\!\to\!\tau_s\\ \Gamma;x_i^1:\tau_1; \ldots; x_i^n:\tau_n\vdash t_i\irefine\langle\tau,\generatecolor{\alpha_2}\rangle
			}
			{ \Gamma\vdash \texttt{\textbf{match}}~e~\texttt{\textbf{with}}~|_i~\texttt{C}_i~(x_i^1 \ldots x_i^n)\mapsto t_i\irefine\langle\tau,\generatecolor{\alpha}\rangle}
		}
		
	\end{mathpar}
	\vspace{-5mm}
	\caption{Typing rules of I-terms}
	\label{Fi:rulesI}
\end{figure*}

\begin{figure*}[tb]
	\scriptsize
	\begin{mathpar}
		\boxed{\Gamma\vdash \gamma\share\gamma_1|\gamma_2}\hspace{5mm}
		{\inferrule*[right=\textsc{S-Pot}]
			{c_1,c_2\ge0\qquad c_1+c_2\le c}
			{ \Gamma\vdash \generatecolor{[c,\phi]}\share\generatecolor{[c_1,\phi]}~|~\generatecolor{[c_2,\phi]}}
		}\hspace{5mm}
		{\inferrule*[right=\textsc{S-Refl}]
			{\ }
			{ \Gamma\vdash \generatecolor{[c,\phi]}\share\generatecolor{[c,\phi]}}
		}\\
		{\inferrule*[right=\textsc{S-Anno}]
			{ \forall i.\Gamma\vdash \generatecolor{[c_i,\phi_i]}\share\generatecolor{[c_i^1,\phi_i]}~|~\generatecolor{[c_i^2,\phi_i]}}
			{ \Gamma\vdash \generatecolor{([c_1,\phi_1], \ldots, [c_n,\phi_n];O(\psi))}\share\generatecolor{([c_1^{1},\phi_1], \ldots, [c_n^{1},\phi_n];O(\psi)})~|~\generatecolor{([c_1^{2},\phi_1], \ldots, [c_n^{2},\phi_n];O(\psi))}}
		}\\
		{\inferrule*[right=\textsc{S-Type}]
			{ \Gamma\vdash \alpha\share\alpha_1~|~\alpha_2}
			{ \Gamma\vdash \langle\tau,\alpha\rangle\share\langle\tau,\alpha_1\rangle~|~\langle\tau,\alpha_2\rangle}
		}\hspace{5mm}
		{\inferrule*[right=\textsc{S-Mul}]
			{ \Gamma\vdash \gamma\share\gamma~|~\gamma'\quad\Gamma\vdash \gamma'\share\gamma_2~|~\gamma_3}
			{ \Gamma\vdash \gamma\share\gamma_1~|~\gamma_2~|~\gamma_3}
		}
	\end{mathpar}
	\caption{Sharing rules}
	\label{Fi:share}
\end{figure*}

\mypar{Cost sharing.}
The sharing operator $\alpha\share\alpha_1|\alpha_2$ partitions the recursive-call cost of $\alpha$ into $\alpha_1$  and $\alpha_2$---i.e., the sum of the costs in $\alpha_1$ and $\alpha_2$ equals the cost in $\alpha$. Sharing rules are shown in \figref{share}.   $\textsc{S-Pot}$ shares a single cost $c$ to two costs $c_1$ and $c_2$ such that their sum is no more than $c$. An annotation can be shared to two parts (\textsc{S-Anno}) if every recursive cost $\generatecolor{[c_i,\phi_i]}$ in it can be shared to two parts $\generatecolor{[c_i^{1},\phi_1]}$ and $\generatecolor{[c_i^{2},\phi_2]}$. Finally, annotations can also be shared to more than two parts (\textsc{S-Mul}).

\mypar{Soundness theorem.}

The proof of the soundness theorem uses  the following crucial lemma connecting the recurrence annotations with the actual recurrence relations of functions.
\begin{lemma}
	\label{Lem:rec}
	Given a function $\texttt{fix }f.\lambda x_1..\lambda x_n.t$, if\ $\Gamma\vdash t\irefine \langle\tau,\generatecolor{([c_1,\phi_1],..,[c_m,\phi_m];O(\psi))}\rangle$, then the complexity~$T_f$ of $f$ satisfies the recurrence relation $T(u)\le c_1T(\phi_1)+..+c_mT(\phi_m)+O(\psi)$
\end{lemma}
\begin{proof}

\begin{changebar}
	
	The proof of the lemma consists two parts: 1) the recursive calls are tracked correctly, and 2) the bound on the function body is tracked correctly. We assume that the function $\texttt{fix }f.\lambda x_1..\lambda x_n.t$ can be evaluated with the recurrence parameter $\checkcolor{\mathcal{R}_1\parallel..\parallel\mathcal{R}_m}$, and show that $\generatecolor{([c_1,\phi_1],..,[c_m,\phi_m];O(\psi))}$ over-approximates the recurrence parameter in such a way that, along any path $\checkcolor{\mathcal{R}_i}$, the number of $T(\phi)$ such that $\phi\le \phi_i$ is no more than $c_i$, and the sum of non-$T()$ terms in $\checkcolor{\mathcal{R}_i}$ .

	First, note that if $\Gamma\vdash t::\langle \{B~|~\varphi_t\},\generatecolor{\alpha}\rangle$ and $\config{t,0}\mapsto\config{\varphi'_t,  \checkcolor{\mathcal{R}_t}}$, then $\varphi'_t\implies\varphi_t$ because they are built in the same bottom-up way. The difference is that the type $\varphi_t$ inferred by the type system can be arbitrarily  over-approximated by the \textsc{E-SubType} rule. 
	
	For a recursive application term  $t:=f(e_1,\dots,e_n)$,  its type will be a subtype of some type containing a recursive-call cost $\generatecolor{[1,\phi_k]}$ if  		$\bigwedge_{i=1}^m[y_i/\emph{v}]\varphi_{e_i}\!\Rightarrow\! (\texttt{size}~y_1 \ldots y_m\le[\texttt{size}~\Gamma(\texttt{args})/\emph{u}]\phi_k)$. Similarly, the term $t$ will be evaluated with loose cost $T(\texttt{size}(v_1,..,v_n))$ where the $v_i$'s are bounds of $\varphi'_{e_i}$. Consider  any bounds $v_i$ of $\varphi_{e_i}$; they are also the bounds of $\varphi'_{e_i}$ because $\varphi'_{e_i}\implies\varphi_{e_i}$. So the size expression $[\texttt{size}~\Gamma(\texttt{args})/\emph{u}]\phi_k$ is also the bound of $T(\texttt{size}(v_1,..,v_n))$, which means that $\checkcolor{\mathcal{R}}$ contains one $T()$-term match $\generatecolor{[1,\phi_k]}$.

	For a non-recursive application term $t:=g(e_1,\dots,e_n)$, if it it has the annotation $\generatecolor{O(\psi_t)}$, then it satisfies $$\bigwedge_{i=1}^m[y_i/\emph{v}]\varphi_i\!\Rightarrow\! \left([\texttt{size}_{\texttt{g}}~y_1 \ldots y_m/\emph{u}]\psi_{\texttt{g}}\in O([\texttt{size}~\Gamma(\texttt{args})/\emph{u}]\psi_t)\right).$$ For a similar reason as above, the loose cost $\checkcolor{~[(\texttt{size}_g~v_1..~v_n)/\emph{u}]}$ is also in $O([\texttt{size}~\Gamma(\texttt{args})/\emph{u}]\psi_t)$. Therefore, along any given path, the recurrence annotation matches the recurrence parameter.
	
	In the branching rules \textsc{T-Match} and \textsc{T-If}, all branches share the same annotation $\alpha_2$. That is, the annotation $\alpha_2$ is the \emph{upper bound} of the recurrence parameters of all branches.  So the recurrence annotation over-approximates all the paths of the recurrence parameter.
	
	Now we have that with the annotation $\langle\tau,\generatecolor{([c_1,\phi_1],..,[c_m,\phi_m];O(\psi))}\rangle$, the function $f$ can be evaluated with some recurrence parameter $\checkcolor{\mathcal{R}_1\parallel..\parallel\mathcal{R}_m}$ such that for any 
	$\checkcolor{\mathcal{R}_i}:=\sum_{j=1}^lc'_jT(\phi'_j)+O(\psi')$, we have $l<m$, $c_j'<c_j$, $\phi'_j\le\phi_j$ and $\psi'\in O(\psi)$.  Then, according to the \theoref{looseIsSound} the complexity of $f$ must satisfy
	$$T(\emph{u})\le \sum_{j=1}^lc'_jT(\phi'_j)+O(\psi')\le \sum_{j=1}^mc_jT(\phi_j)+O(\psi).$$	
\end{changebar}	
\end{proof}

\begin{theorem}[Soundness of type checking]
	Given a function $\texttt{fix }f.\lambda x_1 \ldots \lambda x_n.t$ and an environment $\Gamma$, if $\Gamma\vdash \texttt{fix }f.\lambda x_1 \ldots \lambda x_n.t:: \langle\tau,\generatecolor{O(\psi)} \rangle$, then the complexity of $f$ is bounded by $\psi$.
\end{theorem}
\begin{proof}This theorem can be proved by combining \lemref{rec} and theorems used in the \tableref{pattern}. For example, if a function satisfies $T(\emph{u})\le T(\floor{\frac{\emph{u}}{2}})+O(1)$, than with Master Theorem\cite{bentley1980general} we have $T(\emph{u})\in O(\log \emph{u})$.
\end{proof}

\section{Synthesis Algorithm}
\label{Se:aSynthesis}

\begin{theorem}[Soundness of synthesizing]
	Given a goal type $\langle\tau,\generatecolor{O(\psi)}\rangle$ and an environment $\Gamma$, if a term $\texttt{fix }f.\lambda x_1..\lambda x_n.t$ is synthesized by \name, then the complexity of $f$ is bounded by $\psi$.
\end{theorem}
	\begin{proof}
		The theorem can be proved by showing that terms produced by the synthesis algorithm are all well-typed. First, E-terms must be well typed because the algorithm do enumerate-and-check. Only checked E-terms can be produced. 
		
		Branching terms  are also well typed because we construct goal types for their sub-terms using the premises of the typing rules. By induction hypothesis, their sub-terms are well-typed and so they are well-typed. 
	
	\end{proof}

\begin{changebar}
	\section{\name with Higher-Order Functions}
	\label{Se:hof}
	Recall that the type system of \name shown in \sectref{typeSystem} does not support higher-order functions.  That is, no argument of a function is allowed to be of an arrow type.  Inferring the resource usages of higher-order functions is challenging for two reasons. First, the resource usages of function arguments $g$ can be unknown, in which case the resource usages of applications of $g$ will also be  unknown. Second, the behavior of function arguments $g$ can be unbounded. Hence, the resource usages of nested applications $f(g(...),...)$ can also be unbounded when the resource usage of $f$ grows along with its arguments. 
	
	For example,  the following program is a higher-order function. It takes as input a function argument $g$ and an integer-list argument $xs$, and constructs a new list as output by applying an auxiliary function \texttt{square} and the function argument $g$ to each element in $xs$. When the resource usage of $g$ is unknown, the resource usage of the application $(g~x)$ is also unknown. Also, suppose we assume that the resource usage of \texttt{square} is linear in its argument. In that case, the resource usage of the application $\texttt{square}~(g~ x)$ is unbounded because the value of its argument $(g~x)$ is unbounded.
	\begin{eqnarray*}
		\texttt{map\_square}=\lambda g.\lambda xs.~\texttt{match}~&xs&~\texttt{with}\\
		\texttt{Nil}~ &\to&~\texttt{Nil}\\
		\texttt{Cons}~x~xt~&\to&~\texttt{Cons}~(\texttt{square}~(g~ x))~(\texttt{map\_square}~g~xt)
	\end{eqnarray*}
	
	Although \name does not support higher-order functions in general, we can extend it to support programs with higher-order functions in practice by introducing four restrictions on target programs. First, we assume that the resource usage of each function argument $g$ is a constant, i.e., $g:\langle \tau,\generatecolor{O(1)} \rangle$. Second, function arguments in recursive calls in the synthesized programs are the same as the top-level function's function arguments. For example, in the body of a higher-order function $\texttt{fix}~f.\lambda g\lambda x\lambda y.t$, all recursive application terms must be of form $f(g,\_,\_)$ where each $\_$ can be any well-typed term. Third, we assume that the behavior of  function arguments does not affect the asymptotic resource usage of  higher-order functions. To satisfy this restriction, we want to avoid nested application terms where the outer functions have non-constant resource usage and the value of arguments of the outer function depends on some function arguments. 
	Finally, function arguments cannot appear in size functions. 
	
	In the rest of this section, we first introduce the extensions to \name's type system and then formally state the restrictions we introduced.
	
	\mypar{Extended Syntax and Types.} The extended syntax of the surface language contains two new rules: 1) an E-term can be a function term, which means that arguments of application terms can be function terms, and 2) a function term can be a lambda term. 
	\[
	\begin{array}{lrclrcl}
	\text{E-term}	 &e&::=& f \qquad
	\text{Function term}	 &f&::=& \lambda x.t
	\end{array}
	\]
	
	The extended type system contains a new kind of arrow type 
	$$\tau::=x_1\!:\!\gamma_1\!\to \ldots \to\!x_n\!:\!\gamma_n\!\to\!y:\tau_y,$$
	which extends standard arrow types by allowing arguments to be annotated types. The idea of the extended arrow types is that arguments can be of function types with annotated resource usage. Recall that with the first restriction, we assume that all recurrence annotations in the higher-order arrow type are \generatecolor{O(1)}. That is, the resource usages of function arguments are always constant.

	\mypar{Restriction on the Synthesis algorithm.}
	With the extended type system, we also modify the synthesis algorithm to prune E-terms that breaks the second or third restriction mentioned above.
	
	To support the second restriction (i.e., that we need to call the same function arguments in recursive calls), the synthesis algorithm first stores the function arguments of the top-level functions. Later, when a recursive call is enumerated, the synthesizer checks whether it calls the same function arguments, and rejects the candidate if it does not.
	
	To support the third restriction (i.e., that the behavior of function arguments should not affect the resource usage), the synthesis algorithm avoids enumerating nested application terms where the resource usage of the outer application depends on the value of an inner application term that calls a function argument.
	
	\mypar{Evaluation.} We evaluated \name on 6 benchmarks with higher-order functions after extending the implementation to support the restrictions presented above. The result (HOF in \tableref{results}) was that \name solved 5 of them.
	
\end{changebar}

\section{Detailed Evaluation}
\label{Se:aEval}
\begin{changebar}
\tableref{results} shows the  evaluation results of \synquid, \resyn, and \tool on benchmarks that can be encoded as \tool benchmarks. In the table,
\begin{itemize}
	\item 
	$ T$ denotes running time; $ B$ denotes the given resource bounds;
	TO denotes a timeout;
	\item benchmarks that cannot be encoded by some tools are shown as -;
	\item Rec. rel. represents the recurrence-relation pattern \tool chose to use;
	\item  C=C-finite sequence. M=Master Theorem. A=Akra-Bazzi method. T=the tree recurrence we introduced in \sectref{extension}. N=non-recursive.
\end{itemize} 
\end{changebar}

\begin{table*}[!t]
	\caption{ Evaluation results of \synquid, \resyn, and \tool.}
	\scriptsize
	\centering
	\begin{tabular}{c|c|c|c|c|c|c|c|}
		&	 \multirow{2}{*}{\bf Problem} & \multicolumn{1}{c|}{\bf \synquid}&  \multicolumn{2}{c|}{\bf \resyn}&   \multicolumn{2}{c|}{\bf \tool} &{Rec.} \\
		%& \multirow{2}{*}{\bf extra col} & \multirow{2}{*}{\bf extra col} \\
		%multicolumn{1}{C{10mm}}{\bf Time} &  \multicolumn{2}{|c}{\bf Time single line} [sec]\\
		&   &     { $ T$(sec)} & { $ B$}  &  { $ T$(sec)} & { $ O(B)$}  &  { $ T$(sec)}&{rel.}\\
\hline List	&	is empty 	&	0.64	&	$0$	&	0.65	&	$1$	&	0.64	&	N	\\
&	is member 	&	0.92	&	$|xs|$	&	0.89	&	$|xs|$	&	0.84	&	C	\\
&	duplicate each element 	&	0.87	&	$|xs|$	&	1.63	&	$|xs|$	&	0.93	&	C	\\
&	replicate 	&	1.02	&	$n$	&	8.31	&	$n$	&	1.30	&	C	\\
&	append two lists 	&	0.94	&	$|xs|$	&	3.70	&	$|xs|$	&	1.95	&	C	\\
&	concatenate list of lists 	&	0.93	&	-	&	-	&	$|xss|$	&	0.97	&	C	\\
&	take firstn elements 	&	1.03	&	$n$	&	7.75	&	$n$	&	1.31	&	C	\\
&	drop firstn elements 	&	0.89	&	$n$	&	40.82	&	$n$	&	11.51	&	C	\\
&	delete value 	&	0.90	&	$|xs|$	&	2.04	&	$|xs|$	&	1.30	&	C	\\
&	zip 	&	0.91	&	$|xs|$	&	2.44	&	$|xs|$	&	1.22	&	C	\\
&	i-th element 	&	0.70	&	$|xs|$	&	1.01	&	$|xs|$	&	0.97	&	C	\\
&	index of element 	&	0.97	&	$|xs|$	&	1.76	&	$|xs|$	&	1.29	&	C	\\
&	insert at end 	&	1.06	&	$|xs|$	&	1.65	&	$|xs|$	&	1.04	&	C	\\
&	reverse 	&	1.10	&	$|xs|$	&	1.49	&	$|xs|$	&	1.09	&	C	\\
\hline Unique	&	insert 	&	1.01	&	$|xs|$	&	2.24	&	$|xs|$	&	2.89	&	C	\\
list	&	delete 	&	0.81	&	$|xs|$	&	1.61	&	$|xs|$	&	2.13	&	C	\\
&	remove duplicates 	&	0.74	&	-	&	-	&	$|xs|^2$	&	3.68	&	C	\\
&	compress	&	2.62	&	$|xs|$	&	10.25	&	$|xs|$	&	7.91	&	C	\\
&	integer range 	&	4.83	&	$size$	&	206.19	&	$size$	&	7.42	&	C	\\
\hline Strictly	&	insert 	&	1.24	&	$|xs|$	&	4.92	&	$|xs|$	&	1.49	&	C	\\
sorted list	&	delete 	&	0.75	&	$|xs|$	&	1.92	&	$|xs|$	&	1.24	&	C	\\
&	intersect 	&	4.45	&	$|xs|+|ys|$	&	7.17	&	$|xs|+|ys|$	&	8.91	&	C	\\
\hline Sorting	&	insert (sorted) 	&	0.90	&	$|xs|$	&	3.92	&	$|xs|$	&	2.16	&	C	\\
&	insertion sort 	&	0.73	&	-	&	-	&	$|xs|^2$	&	6.03	&	C	\\
&	extract minimum 	&	2.45	&	$|xs|$	&	28.66	&	$|xs|$	&	10.09	&	C	\\
&	quick sort	&	6.76	&	-	&	-	&	$|xs|^2$	&	39.29	&	C	\\
&	selection sort 	&	1.84	&	-	&	-	&	$|xs|^2$	&	3.42	&	C	\\
&	balanced split 	&	4.09	&	$|xs|$	&	28.59	&	$|xs|$	&	9.59	&	C	\\
&	merge	&	7.14	&	-	&	-	&	$|xs|+|ys|$	&	37.71	&	C	\\
&	merge sort 	&	6.89	&	-	&	-	&	$|xs|\log |xs|$	&	69.63	&	A	\\
&	partition 	&	5.77	&	$|xs|$	&	40.55	&	$|xs|$	&	10.77	&	C	\\
&	append with pivot 	&	1.33	&	-	&	-	&	$|xs|$	&	1.96	&	C	\\
\hline Tree	&	is member 	&	0.97	&	$2|t|$	&	8.88	&	$|t|$	&	5.47	&	T	\\
&	node count 	&	0.84	&	$2|t|$	&	8.94	&	$|t|$	&	2.94	&	T	\\
&	preorder 	&	1.07	&	$2|t|$	&	7.42	&	$|t|$	&	5.69	&	T	\\
\hline BST	&	is member	&	0.75	&	$2|t|$	&	1.83	&	$|t|$	&	1.54	&	T	\\
&	insert	&	2.14	&	$|t|$	&	12.87	&	$|t|$	&	6.56	&	T	\\
&	delete	&	9.89	&	$2|t|$	&	98.04	&	$|t|$	&	24.20	&	T	\\
&	BST sort	&	4.23	&	$3|t|$	&	54.47	&	$|t|$	&	6.28	&	T	\\
\hline Binary	&	is member 	&	1.45	&	$2|t|$	&	0.97	&	$|t|$	&	2.09	&	T	\\
Heap	&	insert 	&	2.01	&	$|t|$	&	11.89	&	$|t|$	&	4.02	&	T	\\
&	1-element constructor	&	0.90	&	$1$	&	2.11	&	$1$	&	1.29	&	N	\\
&	2-element constructor 	&	1.15	&	$2$	&	2.63	&	$1$	&	1.04	&	N	\\
&	3-element constructor 	&	5.34	&	$3$	&	62.69	&	$1$	&	5.06	&	N	\\
\hline AVL	&	rotate left	&	9.84	&	-	&	-	&	$1$	&	9.28	&	N	\\
&	rotate right	&	28.67	&	-	&	-	&	$1$	&	30.44	&	N	\\
&	balance	&	3.95	&	-	&	-	&	$1$	&	4.22	&	N	\\
&	insert 	&	3.92	&	-	&	-	&	$\log|t|$	&	13.36	&	M	\\
&	delete	&	7.99	&	-	&	-	&	$\log|t|$	&	13.82	&	M	\\
&	extract minimum 	&	8.22	&	-	&	-	&	$\log|t|$	&	12.26	&	M	\\
\hline RBT	&	balance left	&	12.63	&	-	&	-	&	$1$	&	11.15	&	N	\\
&	balance right	&	14.81	&	-	&	-	&	$1$	&	15.70	&	N	\\
&	insert	&	3.00	&	-	&	-	&	$\log|t|$	&	9.68	&	M	\\
\hline User	&	make address book 	&	6.50	&	-	&	-	&	$|adds|$	&	4.89	&	C	\\
&	merge address books 	&	1.50	&	-	&	-	&	$1$	&	1.64	&	N	\\
\hline HOF & map & 0.03 &  |xs| &  0.33  & |xs| &  0.58  & C\\
& zip with function & 0.07 &  |xs| & 0.82 & |xs| & 1.11  & C\\
& foldr & 0.10 &  |xs| & 1.88 & |xs| & 2.57  & C\\
& length using fold & 0.03 &  |xs| & 0.67 & |xs| & 0.56  & N\\
& append using fold & 0.04 &  |xs| & 0.34 & |xs| & 0.72 & N\\
\hline \resyn only	&	triple-1	&	1.01	&	$2|xs|$	&	2.93	&	$|xs|$	&	1.75	&	C	\\
&	triple-2	&	1.01	&	$2|xs|$	&	6.16	&	$|xs|$	&	1.45	&	C	\\
&	concat list of lists	&	1.30	&	$|xss|$	&	9.68	&	$|xss|$	&	1.79	&	C	\\
&	common	&	2.57	&	$|ys|+|zs|$	&	40.77	&	$|ys|+|zs|$	&	57.55	&	C	\\
&	list difference	&	1.36	&	$|ys|+|zs|$	&	419.23	&	$|ys|+|zs|$	&	41.88	&	C	\\
&	insert	&	1.18	&	$\texttt{numgt}(x,xs)$	&	48.82	&	$\texttt{numgt}(x,xs)$	&	3.76	&	C	\\
&	range	&	TO	&	$hi-lo$	&	128.8	&	$hi-lo$	&	7.63	&	C	\\
&	compare	&	1.02	&	$|xs|+|ys|$	&	3.78	&	$|xs|+|ys|$	&	8.32	&	C	\\
\hline\tool only	&	binary search	&	1.53	&	-	&	-	&	$\log |xs|$	&	5.20	&	M	\\
&	product	&	1.09	&	-	&	-	&	$\log x$	&	1.37	&	M	\\
&	binary search'	&	1.64	&	-	&	-	&	$\log |xs|$	&	24.68	&	M	\\
&	product'	&	0.98	&	-	&	-	&	$\log x$	&	14.13	&	M	\\
&	merge sort'	&	TO	&	-	&	-	&	$|xs|\log |xs|$	&	75.73	&	A	\\
		%		\hline
		%		\parbox[t]{-11mm}{\multirow{11}{*}{\rotatebox[origin=c]{90}{\textsc{LimitedPlus}}}}
	\end{tabular}
	\label{Ta:results}
	\vspace{-2mm}
\end{table*}  

\end{document}